\documentclass[journal]{IEEEtran}

\ifCLASSINFOpdf
  \usepackage[pdftex]{graphicx}
  \graphicspath{{../pdf/}{../jpeg/}}
  \DeclareGraphicsExtensions{.pdf,.jpeg,.png}
\else
  \usepackage[dvips]{graphicx}
\fi

\usepackage{cite}
\usepackage[utf8]{inputenc}
\usepackage{amssymb}
\usepackage{amsmath}
\usepackage{amsthm}
\usepackage{mathtools}
\usepackage{optidef}
\usepackage{hyperref}
\usepackage{algorithm}
\usepackage{algpseudocode}
\usepackage{array}
\usepackage{booktabs}
\usepackage{todonotes}
\newtheorem{theorem}{Theorem}
\newtheorem{corollary}{Corollary}
\usepackage{mathtools}
\DeclarePairedDelimiter{\ceil}{\lceil}{\rceil}
\newcommand\copyrighttext{%
  \footnotesize \textcopyright 2020 IEEE. Personal use of this material is permitted.  Permission from IEEE must be obtained for all other uses, in any current or future media, including reprinting/republishing this material for advertising or promotional purposes, creating new collective works, for resale or redistribution to servers or lists, or reuse of any copyrighted component of this work in other works.}
\newcommand\copyrightnotice{%
\begin{tikzpicture}[remember picture,overlay]
\node[anchor=south,yshift=0pt] at (current page.south) {\fbox{\parbox{\dimexpr\textwidth-\fboxsep-\fboxrule\relax}{\copyrighttext}}};
\end{tikzpicture}%
}

\ifCLASSOPTIONcompsoc
  \usepackage[caption=false,font=normalsize,labelfont=sf,textfont=sf]{subfig}
\else
  \usepackage[caption=false,font=footnotesize]{subfig}
\fi

\begin{document}

\title{Vehicle Redistribution in Ride-Sourcing Markets using Convex Minimum Cost Flows}

\author{Renos~Karamanis, Eleftherios~Anastasiadis, Marc~Stettler
        and~Panagiotis~Angeloudis
\thanks{R. Karamanis, E. Anastasiadis, M. Stettler and P. Angeloudis are with the Department
of Civil and Environmental Engineering, Imperial College London,
UK, e-mail: renos.karamanis10@imperial.ac.uk.}
}

\markboth{}%
{Karamanis \MakeLowercase{\textit{et al.}}: Vehicle Redistribution using Convex Minimum Cost Flows}

\IEEEtitleabstractindextext{
\begin{abstract}
Ride-sourcing platforms often face imbalances in the demand and supply of rides across areas in their operating road-networks. As such, dynamic pricing methods have been used to mediate these demand asymmetries through surge price multipliers, thus incentivising higher driver participation in the market. However, the anticipated commercialisation of autonomous vehicles could transform the current ride-sourcing platforms to fleet operators. The absence of human drivers fosters the need for empty vehicle management to address any vehicle supply deficiencies. Proactive redistribution using integer programming and demand predictive models have been proposed in research to address this problem. A shortcoming of existing models, however, is that they ignore the market structure and underlying customer choice behaviour. As such, current models do not capture the real value of redistribution. To resolve this, we formulate the vehicle redistribution problem as a non-linear minimum cost flow problem which accounts for the relationship of supply and demand of rides, by assuming a customer discrete choice model and a market structure. We demonstrate that this model can have a convex domain, and we introduce an edge splitting algorithm to solve a transformed convex minimum cost flow problem for vehicle redistribution. By testing our model using simulation, we show that our redistribution algorithm can decrease wait times by more than 50\%, increase profit up to 10\% with less than 20\% increase in vehicle mileage. Our findings outline that the value of redistribution is contingent on localised market structure and customer behaviour.
\end{abstract}
}
\IEEEoverridecommandlockouts

\maketitle
\copyrightnotice

\IEEEpubidadjcol
\IEEEdisplaynontitleabstractindextext
\IEEEpeerreviewmaketitle

\section{Introduction} \label{sec: I}
\IEEEPARstart{R}{ide}-sourcing companies, also referred to as Transportation Network Companies (TNCs) gradually dominated the pre-existing taxi market over the past decade, with evidence of their immense success stipulated in market analytics of urban transport data such as in New York City \cite{Schneider2020}. TNCs often encounter imbalances in the supply and demand for rides. Such imbalances can increase customer wait times in areas where there is an under-supply of drivers, thereby decreasing the quality of service and the popularity of the platform. To mediate this effect, TNCs apply dynamic pricing strategies, usually in the form of variable surge pricing multipliers \cite{chen2016dynamic}. These dynamic pricing strategies by design motivate drivers to redistribute to under-served areas and suppress demand from customers whom their willingness to pay is exceeded \cite{banerjee2015pricing}. 

The anticipated launch of autonomous vehicles in TNC services to cut operational costs could transform TNCs from matching platforms to fleet operators having complete control of the supply \cite{qiu2018dynamic}. In such a scenario, currently implemented dynamic pricing strategies would still suppress demand \cite{Karamanis2018}, but TNCs, as fleet owners, would need to decide vehicle redistribution operations. Generally, in the absence of drivers entering the market proactively by knowing historical surge pricing patterns \cite{chen2016dynamic}, autonomous vehicle ride-sourcing operators would need to manage their fleet effectively, to alleviate any asymmetries of demand across road-networks.

Fleet management, and especially empty vehicle redistribution, although not prevalent in ride-sourcing markets, has been an established practice in shared mobility (bike and car-sharing) \cite{DELLAMICO20147, NOURINEJAD201598}. In existing shared mobility schemes, vehicle redistribution is carried out by dedicated staff. In the case of autonomous ride-sourcing, platforms would be able to instruct vehicles to self-relocate, thereby avoiding dedicated staff costs. TNCs could also proactively decide vehicle redistribution operations by exploiting the diverse area of predictive algorithms and existing data.

Nonetheless, this seamless autonomous vehicle relocation would endure mileage costs. Besides, increased fleet mileage can induce externalities such as congestion subject to fleet adoption rates \cite{levin2017general}. Furthermore, as ride-sourcing markets are competitive, and travellers encounter alternative options, redistributed vehicles in an area are not guaranteed an assignment. Consequently, vehicle redistribution models which take account of relocation costs, local market structure and travel behaviour are paramount in assessing the value of vehicle redistribution to autonomous ride-sourcing platforms.

Research on taxi economics, such as the work in \cite{Yang2002}, has been seminal in influencing the implementation of spatio-temporal characteristics when modelling taxi markets. The authors in \cite{Yang2002} did assume the flow of taxis in neighbouring areas; however, the work focused on evaluating system performance metrics for regulatory frameworks. Later studies also considered optimizing redistribution in shared mobility schemes such as bike-sharing or car-sharing\footnote{For a comprehensive review of vehicle redistribution algorithms for car-sharing we refer readers to \cite{ILLGEN2019193}} fleets \cite{DELLAMICO20147, NOURINEJAD201598, nair2011fleet, pfrommer2014dynamic}. 

The structural differences\footnote{In traditional shared mobility schemes vehicles are usually parked at designated stations, and dedicated staff performs the relocation. Also, vehicles are usually booked in advance.} between ride-sourcing and traditional vehicle sharing schemes prohibited the direct application of such models in the ride-sourcing market. Vehicle redistribution for ride-sourcing/taxi markets became popular alongside the concept of autonomy. As a consequence, researchers investigated vehicle redistribution in simulation studies implementing simplified redistribution heuristics for shared autonomous vehicles \cite{fagnant2014travel, levin2017general}. These studies were critical in identifying the extra mileage and congestion, respectively, when redistributing empty shared autonomous vehicles.

Recent research focused on identifying redistribution strategies for ride-sourcing operations using demand predictions and linear integer programming, queuing theoretical models or machine learning methods to identify relocation strategies. In \cite{zhang2016control} a closed Jackson network was used to simulate an autonomous mobility-on-demand (MoD) service with passenger loss. The authors then solved the vehicle rebalancing problem using a linear program. Using a queuing-theoretical implementation, they showed that congestion effects due to rebalancing could be avoided. The authors in \cite{iglesias2019bcmp} replaced the Jackson network with a Baskett–Chandy–Muntz–Palacios (BCMP) queuing network model and considered vehicle charging operations.

The authors in \cite{wen2017rebalancing} and \cite{lin2018efficient} used reinforcement learning to identify rebalancing actions in an MoD scheme and ride-sourcing platform respectively and showed that their methods achieve effective rebalancing strategies. A fluid-based optimization problem on a queuing network was used in \cite{braverman2019empty} to identify an optimal routing policy with an upper bound for empty car routing in ride-sharing systems. The study in \cite{YU2019114} considered a Markov decision process for the problem of vacant taxi routing with e-hailing. The authors solved their model using an iterative algorithm to maximise the expected long-term profit over a working period. 

The authors in \cite{alonso2017predictive} and \cite{wallar2018vehicle} used predictive algorithms to estimate incoming requests and an integer programming model to assign idle vehicles to clustered regions. Their models were tested in a simulation of the New York City taxi data and achieved a significant reduction of waiting times. Model predictive control for vehicle redistribution was also utilised in \cite{iglesias2018data} to utilise short term estimations of customer demand. The model in \cite{iglesias2018data} achieved a significant reduction in waiting times when tested in simulation using Didi data.

The majority of relevant studies on vehicle redistribution (Table \ref{tab:lit-review}) do not consider acumen in customer behaviour. The two main approaches are, assuming unassigned customers abort the platform immediately (passenger loss) or setting a maximum customer wait time (maximum wait), after which all customers abort the service. However, in realistic ride-sourcing implementations, potential customers would encounter alternative travel options based on local competition, thereby deciding their mode choice based on several factors \cite{yang2020integrated}. Consequently, oversimplified models of customer behaviour or nonexistent representation of competition do not reflect the real value and cost of redistribution, nor do they account for the notion of diminishing returns, since vehicle relocations do not necessarily result in guaranteed customers. 

To address this literature gap, we present a mathematical programming formulation that accounts for alternative transport services offered and customer choice behaviour. We model the derived vehicle redistribution problem as a non-linear minimum cost flow problem and prove that the model can have an optimal solution in a convex domain. We transform the non-linear model to the convex minimum cost flow problem and solve it using an edge-splitting pseudo-polynomial algorithm.

Our contribution is summarised as follows:
\begin{enumerate}
    \item We incorporate customer behaviour and local competition to account for the notion of diminishing returns in the vehicle redistribution problem.
    \item We model the vehicle redistribution problem as a Non-Linear Minimum Cost Flow problem.
    \item We derive a convex space for the problem and transform it into a Convex Minimum Cost Flow Problem, which is solved using a pseudo-polynomial algorithm edge splitting algorithm.
\end{enumerate}

The remainder of this paper is structured as follows: In section \ref{Sec: II}, we outline the structure of our proposed vehicle redistribution model as a non-linear minimum cost flow model. We then prove the existence of a convex region and present our edge-splitting solution algorithm. In section \ref{sec: III}, we test our redistribution methodology in an agent-based model by using taxi data from New York City. Conclusions and recommendations for further work are provided in section \ref{sec: IV}. 

\begin{table}[]
\centering
\caption{Relevant studies on vehicle redistribution}
\label{tab:lit-review}
\begin{tabular}{@{}llll@{}}
\toprule
Study                              & Method                                                                           & \begin{tabular}[c]{@{}l@{}}Redistribution \\ Cost\end{tabular}     & \begin{tabular}[c]{@{}l@{}}Customer \\ Behavior\end{tabular} \\ \midrule
\cite{zhang2016control, iglesias2019bcmp} & Queuing theory                                                                   & Congestion                                                         & \begin{tabular}[c]{@{}l@{}}Passenger loss\end{tabular}    \\
\cite{wen2017rebalancing}                 & \begin{tabular}[c]{@{}l@{}}Reinforcement \\ Learning\end{tabular}                & Distance-based                                                     & \begin{tabular}[c]{@{}l@{}}Passenger loss\end{tabular}    \\
\cite{lin2018efficient}                   & \begin{tabular}[c]{@{}l@{}}Reinforcement \\ Learning\end{tabular}                & Fuel-based                                                         & Not-specified                                                \\
\cite{braverman2019empty}                 & Queuing theory                                                                   & Time-based                                                         & \begin{tabular}[c]{@{}l@{}}Passenger loss\end{tabular}    \\
\cite{YU2019114}                          & \begin{tabular}[c]{@{}l@{}}Markov decision \\ process\end{tabular}               & Time-based                                                         & Not-specified                                                \\
\cite{alonso2017predictive}               & \begin{tabular}[c]{@{}l@{}}Demand prediction,\\ integer programming\end{tabular} & Passenger delay                                                    & \begin{tabular}[c]{@{}l@{}}Maximum wait\end{tabular}      \\
\cite{wallar2018vehicle}                  & \begin{tabular}[c]{@{}l@{}}Demand prediction,\\ integer programming\end{tabular} & Time-based                                                         & \begin{tabular}[c]{@{}l@{}}Maximum wait\end{tabular}      \\
\cite{iglesias2018data}                   & \begin{tabular}[c]{@{}l@{}}Demand prediction,\\ integer programming\end{tabular} & \begin{tabular}[c]{@{}l@{}}Time and \\ Distance-based\end{tabular} & No passenger loss                                                      \\ \bottomrule
\end{tabular}
\end{table}

\section{Methodology} \label{Sec: II}
\subsection{Vehicle Redistribution Problem} \label{RAP}
We consider an autonomous vehicle ride-sourcing fleet operator, opposed with the problem of identifying an allocation of vehicles to various operations to minimise the fleet's operational cost. The fleet operator identifies allocations of the vehicles at regular decision periods and operates in an urban road network split into different clusters.

Immediately before the allocation decision, the operator identifies the vehicle counts in each cluster, including the vehicles soon to be located in each area. At each decision epoch, the fleet operator needs to allocate the vehicles in each cluster into three possible operational states; available for trip allocation, empty redistribution, or idle. 

Vehicles assigned for trip allocation are immediately available for trip requests originating from their existing cluster, and their number depends on demand estimates for the commencing period. Empty redistribution refers to vehicles allocated for empty travel to other clusters to satisfy demand estimates for the commencing and subsequent periods. Finally, idle vehicles remain inactive in their initial cluster for the commencing period and act as reserve capacity for the fleet if required.

Consequently, vehicles are allocated from their initial state and clusters into the various operations at the beginning of the decision period and end up in updated states and clusters for the subsequent period. For convenience, we refer to the updated vehicle states as resulting states and the operational states as decision states.

We define the set of road network clusters $J$ and assume the fleet operator has estimates of the total demand $Z^t_{ij}$ from cluster $i$ to cluster $j$ for each $i,j \in J$ and for every time epoch $t \in T$. We assume the mean utility of travel in time epoch $t$ for autonomous ride-sourcing trips from cluster $i$ to cluster $j$ for each $i,j \in J$ is realised using the following generalised cost function:

\begin{equation} \label{eq:1}
    g^t_{ij}(x) = -\bar{v}(w^t_{ij}(x) + r^t_{ij}) -p r^t_{ij} \quad \forall i,j \in J, \forall t \in T
\end{equation}

Where $\bar{v}$ in \eqref{eq:1} is the mean value of time of the ride-sourcing travellers, $w^t_{ij}(x)$ is the average wait time from request to pick up for a trip originating in cluster $i$ and terminating in cluster $j$ at period $t$ for a supply of vehicles $x$, $r^t_{ij}$ is the average travel time from cluster $i$ to cluster $j$ during period $t$ and $p$ is the price per time for the service. 

The studies in \cite{zha2018geometric}, \cite{korolko2018dynamic} and \cite{karamanis2020abm} derived estimates of the pickup wait time which are inversely proportional to the square root of idle vehicles in an area over a time interval. This relationship is consistent across the literature and takes inputs such as the area and average velocity in the network. Nonetheless, the relationship between idle vehicles over a period in a cluster and the total supply of vehicles $x$ is non-trivial and of dynamic nature, as it depends upon the rate of requests and idle vehicle arrivals \cite{karamanis2020abm}. We, therefore, define a function of wait time $w^t_{ij}(x)$, which has the following properties: 

\begin{equation}
    w^t_{ij}(x) = \frac{\beta^t_{i}}{\sqrt{I^t_{ij}(x)}} \quad \forall i,j \in J, t \in T, x \in \mathbb{R}^+\label{eq:2}
\end{equation}

Where $I^t_{ij}(x)$ in equation \eqref{eq:2} is an arbitrary function of the average number of idle vehicles in cluster $i$ available for travel to cluster $j$ with respect to supply $x$ on period $t$. Similar to the literature, $\beta^t_{i}$ is proportional to the square root of the cluster area and inversely proportional to the average velocity in the network. $\beta^t_{i}$ can be determined via calibration with the use of agent-based modelling (ABM) \cite{karamanis2020abm}. By utilising a discrete choice model and assuming constant values for $r^t_{ij}$ and $p^t_{ij}$ in each period, we can calculate the proportion of travellers $q^t_{ij}(x)$ choosing the ride-sourcing fleet as an option to travel from cluster $i$ to cluster $j$ for a period $t$.

\begin{equation} \label{eq:3}
    q^t_{ij}(x) = \frac{e^{g^t_{ij}(x)}}{e^{g^t_{ij}(x)} + \sum_{w \in W} e^{U^{tw}_{ij}}} \quad \forall i,j \in J, \forall t \in T
\end{equation}

Where $W$ in \eqref{eq:3} refers to the set of alternative ride-sourcing options and $U^{tw}_{ij}$ is the mean utility of option $w\in W$ for travelling from cluster $i$ to cluster $j$ in period $t$. As such, the number of travellers $N^t_{ij}(x)$ choosing the ride-sourcing service at period $t$ to travel from cluster $i$ to cluster $j$ for a supply of vehicles $x$ is found using the following equation:

\begin{equation} \label{eq:4}
    N^t_{ij}(x) =  q^t_{ij}(x) Z^t_{ij} \quad \forall i,j \in J, \forall t \in T
\end{equation}

Equation \eqref{eq:3} is a sigmoid function, as we can represent it in the form of the logistic function. Equation \eqref{eq:4} is a scaled version of the sigmoid function in \eqref{eq:3}. Consequently, due to its non-linearity and monotonically increasing nature, the function for the number of travellers choosing the service $N^t_{ij}(x)$ does not necessarily match the supply of vehicles $x$.  

\subsection{Non-Linear Minimum Cost Flow Formulation} \label{NMCF}
The logic defined in equations \eqref{eq:1}-\eqref{eq:4}, represents the aggregated model which the fleet operator uses to estimate trips from and to each cluster. The number of travellers choosing the ride-sourcing service, however, is contingent on the redistribution strategy and the vehicle supply, which shapes the service quality defined in equation \eqref{eq:2}. 

We, therefore, propose a minimum cost flow formulation in the form of resource allocation to solve the vehicle redistribution problem described in section \ref{RAP}. Consider a directed acyclic graph $G=(V, E)$, with $V$ and $E$ representing the sets of graph vertices and edges respectively. The set of vertices $V$ consists of three subsets $A$, $B$, $C$, representing the initial states, decision states and resulting states respectively, such that $V=A \cup B \cup C$. 

For initial state vertices, we consider the numbers of available vehicles at the beginning of epoch $t$ at each cluster. We also subdivide the decision state vertices into the subsets $K$, $L$, $M$ of trip, redistribution and idle states respectively, such that $B=K \cup L \cup M$. Finally, for resulting states, we consider the numbers of available vehicles at the beginning of epoch $t+1$ at each cluster. We associate vertices in each set with the set of road network clusters $J$ such that $|A|=|K|=|L|=|M|=|C|=|J|$.  Consequently, the cardinality $n$ of the set $V$ of vertices is $n=|V|=5|J|$. 

A graph edge $(i,j) \in E$ between two vertices $i,j \in V$, represents the change from state $i$ to state $j$ due to allocation decisions. Figure \ref{fig:1} outlines an example of our proposed resource allocation graph with two clusters. We assume vertices in $A$ have directed edges which connect to the vertices in $B$ only in their respective cluster, with a direction from $A$ to $B$. Consequently, there are three edges from vertices in $A$ to vertices in $B$ for each cluster.

We further assume edges between vertices in $B$ and $C$. Each vertex in $K$ connects to all vertices in $C$. For redistribution edges starting from vertices in $L$ and terminating to vertices in $C$, we exclude edges which start and terminate in the same cluster. Finally, for each cluster, we assume an edge from the corresponding cluster vertex in $M$ to the corresponding cluster vertex in $C$. As such the cardinality $m$ of the set of edges $E$ is $m=|E|=3|J| + |J|^2 + |J|(|J|-1) +|J|=|J|(2|J|+3)$. 

\begin{figure}[]
\centering
\includegraphics[width=0.48\textwidth]{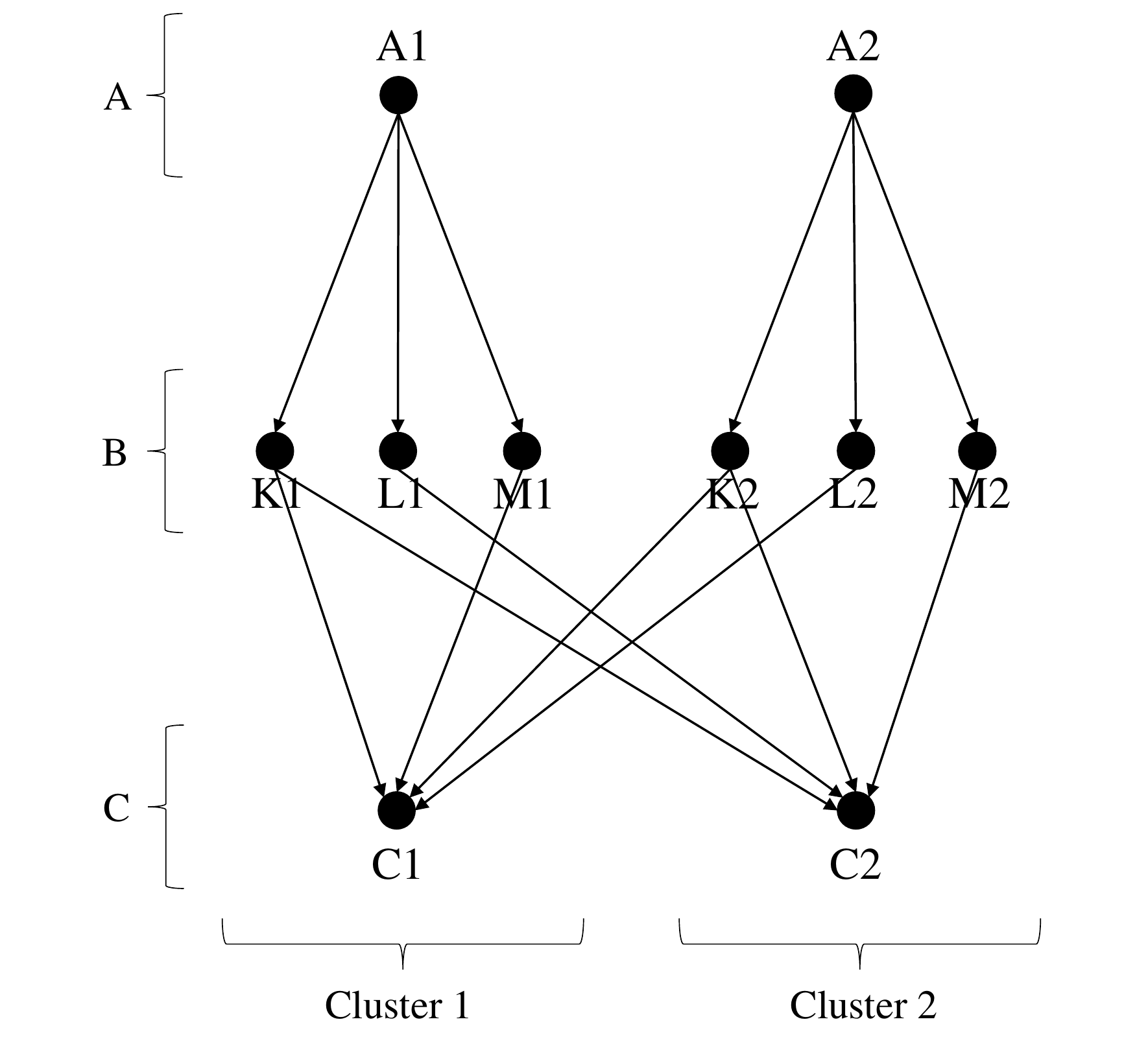}
\caption{Example of the proposed resource allocation graph with 2 clusters.}
\label{fig:1}
\end{figure}

To assist our minimum cost flow formulation, we introduce the following functions on $E$: a lower bound $l_{ij} \geq 0$, a capacity $u_{ij} \geq l_{ij}$ and a cost $c_{ij}$ for each $(i,j) \in E$. Furthermore, we introduce the balance vector function $b:V \rightarrow \mathbb{Z}$ which associates integer numbers with each vertex in $V$. We assume the following holds:

\begin{equation}
    \sum_{v \in V} b(v) = 0 \label{eq:5}
\end{equation}

We denote the flow in graph $G$ as a function $x: E \rightarrow Z^{\geq0}$ on the edge set of $G$, such that the value of the flow on edge $(i,j)$ is $x_{ij}$. The balance vector function $b(v)$ for each $v \in V$ is the difference between the flow in edges of out-degree of $v$ and the flow in edges of in-degree of $v$. As such, the balance vector $b$ is the following function on the vertices:

\begin{equation}
    b(j) = \sum_{i:ji \in E} x_{ji} - \sum_{i:ij \in E} x_{ij} \quad\forall j \in V \label{eq:6}
\end{equation}

We classify vertices with $b(v)>0$ as source vertices, whereas vertices with $b(v)<0$ are sink vertices. Otherwise, if a vertex has $b(v)=0$ we call that vertex balanced. Consequently, a flow $x$ in $G$ is feasible if $l_{ij} \leq x_{ij} \leq u_{ij}$ for all $(i,j) \in E$ and equations \eqref{eq:5} and \eqref{eq:6} hold for all vertices $v \in V$.  Considering our redistribution problem introduced in Section \ref{RAP}, we can regard the vertices $v \in A$ as source vertices since these constitute the initial states of all vehicles in the fleet. In a similar fashion, the vertices in $C$ can be identified as sinks, however; their balance vectors cannot be defined in advance, as certain demand requirements might lead to violation of \eqref{eq:5}. 

The balance vector function values for sink nodes are related to the demand for trips towards each cluster in the road network. As the minimum cost formulation for solving the resource allocation problem would adhere that a feasible flow satisfies equations \eqref{eq:5} and \eqref{eq:6}, we transform our graph from a multi-source, multi-sink to a multi-source, single-sink one. To do so, we introduce the set of vertices $D$, for which $|D|=|J|$ such that there is one vertex from set $D$ in each cluster. We also introduce sink vertex $t$ for which $t \cap J= \emptyset$. Vertices in $D$ denote demand satisfiability for the subsequent period. Consequently we have $V \gets V \cup (D \cup t)$.

Vertices from $C$ are connected with edges to vertices in $D$ in each cluster, to denote that vehicles which terminate their tasks or are idle in a cluster during a time epoch, could be available if needed in the same cluster for the subsequent period. Furthermore, to account for excess vehicles for subsequent demand, we also consider directed edges from all vertices in $C$ to the sink vertex $t$. We transfer the flow from vertices in $D$ to the sink vertex $t$ by including additional directed edges between them. 

Edges starting from vertex sets $K$ to $C$ in each cluster as shown in Figure \ref{fig:1} represent trip edges. If we focus on an individual cluster, the vehicle flow through these edges originates from the same cluster and aims to satisfy the demand for the commencing period. However, depending on geographical proximity, redistributing vehicles from other areas might arrive in the cluster before the end of the commencing period. As a consequence,  redistributing vehicles could be exposed to a portion of the demand originating from the cluster during the commencing period.

The above description implies that within a period, in each cluster, there can be variable supply levels exposed to variable demand portions due to the mixing of redistributing vehicles from different clusters. This behaviour is captured in the formulation of the rebalancing problem in \cite{wallar2018vehicle}. To account for the intra-period redistribution mixing, we introduce additional sets of vertices and edges to the network described in Figure \ref{fig:1} between the vertex sets of $B$ and $C$. Specifically, in each cluster, we add vertices representing the arrival of vehicle flow from redistributing vertices $L$ from other clusters. As a consequence, we add $|J|-1$ vertices in each cluster. These additional vertices are extensions to the vertex subset $K$ since they are trip vertices.

We then connect each additional edge in $K$ with a redistribution vertex in $L$ from other clusters, resulting in $|J|(|J|-1)$ additional edges. We also connect each additional edge in $K$ with all vertices in $C$, resulting in $J^2(|J|-1)$ additional edges. Furthermore, to model the vehicle mixing with vehicles already in each cluster, we need to add edges from the original vertices in $K$, to the additional vertices in $K$. To do so, we identify the sequence of vehicle mixing using a sorted list of the arrival times in each cluster. 

For convenience, in each cluster $i$ in $J$, we denote original vertices in $K$ as $K_i$. As the arrival of vehicles from other clusters has a cumulative effect on the supply in each cluster, we denote the additional vertices in $K$ using the sequence of arrivals. For example, if vehicles from $L_j$ arrive in cluster $i$ before vehicles from $L_k$ for $i,j,k$ in $J$, we denote the vertices corresponding to the mixing of vehicles as $K_{ij}$ and $K_{ijk}$, for redistribution occurring from clusters $j$ and $k$ respectively. Finally, to complete the mixing, in each cluster, we introduce $|J|-1$ edges between the intra-cluster vertices in $K$ according to the sequence of arrivals. Revisiting the above example, in cluster $i$, directed edges are introduced from $K_i$ to $K_{ij}$ and from $K_{ij}$ to $K_{ijk}$. An outline of the transformed graph is shown in Figure \ref{fig:2}. As such, the following equations hold:

\begin{equation}
    b(v)>0 \quad\forall v \in A \label{eq:7}
\end{equation}
\begin{equation}
    b(v)=0 \quad\forall v \in B \cup C \cup D \label{eq:8}
\end{equation}
\begin{equation}
    b(t)<0  \label{eq:9}
\end{equation}

\begin{figure}[]
\centering
\includegraphics[width=0.48\textwidth]{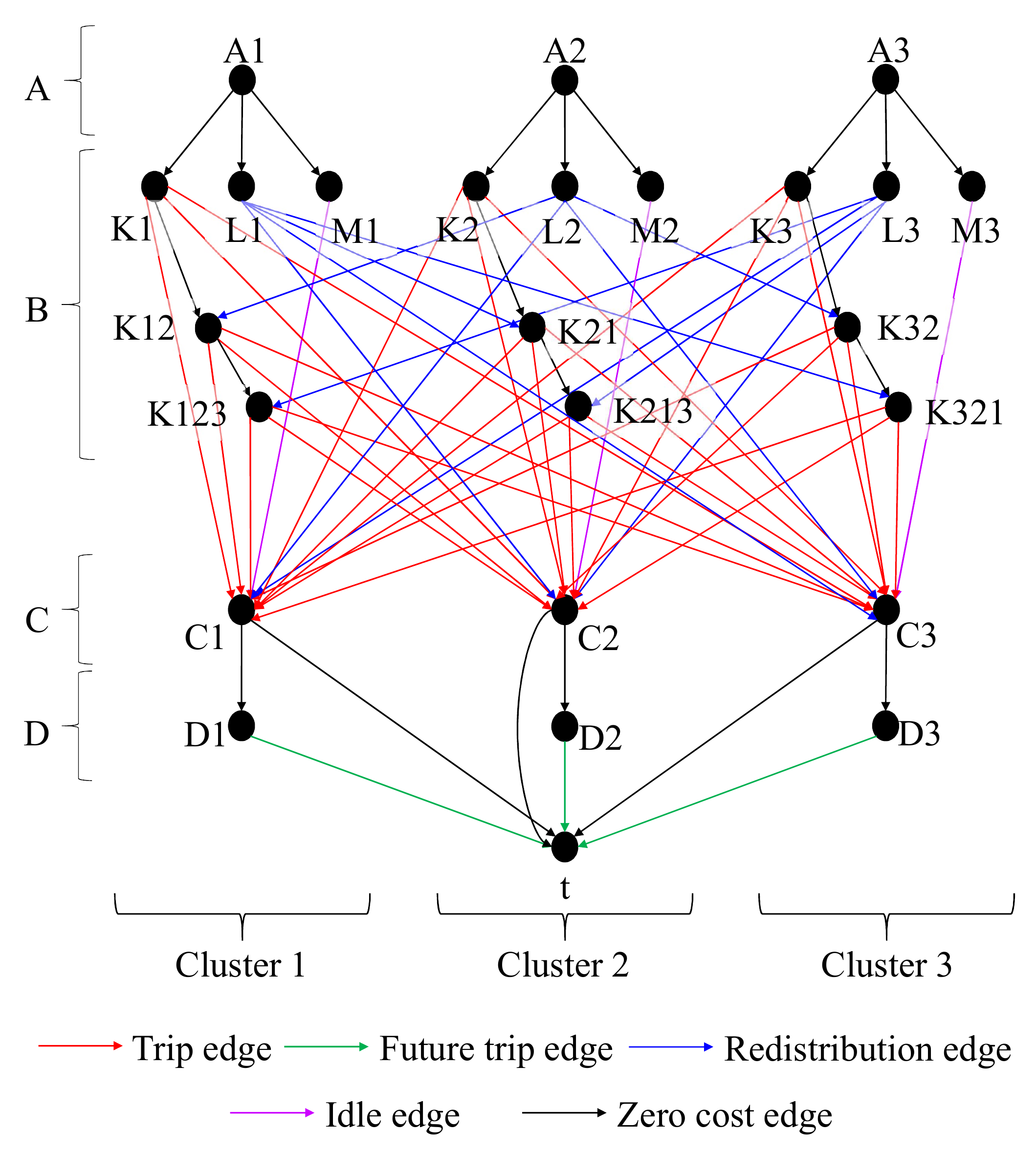}
\caption{Example of the transformed single source/sink resource allocation graph with 3 clusters.}
\label{fig:2}
\end{figure}

To assist our notation, we define a function $n: V \rightarrow J$, which maps vertices of the resource allocation network to clusters in $J$. Furthermore, to simplify set notation for edges, we define the edge sets $\mathcal{A}, \mathcal{B}, \mathcal{C}, \mathcal{D}, \mathcal{E} \in E$. These edge sets represent the trip, future trip, redistribution, idle and zero cost edges respectively, as described in Figure \ref{fig:2}. As such, we define the cost functions for the edges in the graph as follows:

\begin{equation}
    c_{ij}(x_{ij}) =  h_{ij}(x_{ij}) \quad \forall (i,j) \in \mathcal{A} \label{eq:10}
\end{equation}
\begin{equation}
\begin{aligned}
c_{ij}(x_{ij}) = r^1_{n(i)n(j)}C_M x_{ij} \quad \forall (i,j) \in \mathcal{C} \label{eq:11}
\end{aligned}
\end{equation}
\begin{equation}
    c_{ij}(x_{ij}) = C_I x_{ij} \quad \forall (i,j) \in \mathcal{D} \label{eq:12}
\end{equation}
\begin{equation}
    c_{ij}(x_{ij}) =  f_{ij}(x_{ij})  \quad \forall (i,j) \in \mathcal{B}\label{eq:13}
\end{equation}
\begin{equation}
    c_{ij}(x_{ij}) = 0 \quad \forall (i,j) \in \mathcal{E} \label{eq:14}
\end{equation}

Equations \eqref{eq:10}-\eqref{eq:12} define the cost functions for edges directed from decision state vertices $B$ to vertices in $C$. Equation \eqref{eq:10} can have a negative sign, as the revenues from trip allocations with a supply of $x_{ij}$ in the next epoch are subtracted from the cost using function $h_{ij}(x_{ij})$. Parameter $r^1_{n(i)n(j)}$ is the average travel time between the clusters $n(i)$ and $n(j)$ of vertices $i$ and $j$ at the initial epoch as introduced in equation \eqref{eq:1}, and $C_M$ is the cost of a moving vehicle per time. Consequently, equation \eqref{eq:11} defines the cost of redistribution for vehicles. Equation \eqref{eq:12} defines the cost of idle vehicles, with $C_I$ to denote the cost of an idle vehicle per period.

Equation \eqref{eq:13}, is similar to equation \eqref{eq:10} but subtracts the revenues from potential trips in the subsequent period from the costs, depending on vehicle supply, using function $f_{ij}(x_{ij})$. Finally, we set the cost to zero for the remaining edges. Functions $h_{ij}(x_{ij})$ and  $f_{ij}(x_{ij})$  in equations \eqref{eq:10} and \eqref{eq:13} are outlined in the following equations:

\begin{equation}
    \begin{aligned}
        h_{ij}(x_{ij}) =  -\phi_{ij} N^1_{n(i)n(j)}(x_{ij}) p r^1_{n(i)n(j)} + C_M r^1_{n(i)n(j)} x_{ij}\\ \quad\forall(i,j)\in \mathcal{A} \label{eq:15}
    \end{aligned}
\end{equation}

\begin{equation}
   \begin{aligned}
        f_{ij}(x_{ij}) = \sum_{m \in J}\bigg( -N^2_{n(i)m}\bigg(\frac{x_{ij}}{|J|}\bigg) p r^2_{n(i)m}+ C_M r^2_{n(i)m} x_{ij}\bigg)\\
        \forall(i,j)\in \mathcal{B} \label{eq:16}
    \end{aligned}
\end{equation}

As observed in equations \eqref{eq:15} and \eqref{eq:16}, both functions utilise $N^t_{ij}(x)$, which refers to the number of travellers choosing the service given a supply of vehicles $x$ introduced in equation \eqref{eq:4}. Parameter $p$ is the revenue per time for each vehicle as in equation \eqref{eq:1}, $r^t_{ij}$ is the travel time between clusters $i$ and $j$ during epoch $t$ as in equations \eqref{eq:1} and \eqref{eq:11}. To accommodate vehicle mixing from redistribution, we factor the number of travellers $N^t_{ij}(x)$ in equation \eqref{eq:15} by a factor $\phi_{ij} \in [0,1]$ according to the arrival sequence of redistributing vehicles in each cluster. A visual demonstration of how the demand factor $\phi_{ij}$ is calculated for each edge in $\mathcal{A}$ in a cluster is highlighted in Figure \ref{fig:21}.

\begin{figure}[]
\centering
\includegraphics[width=0.48\textwidth]{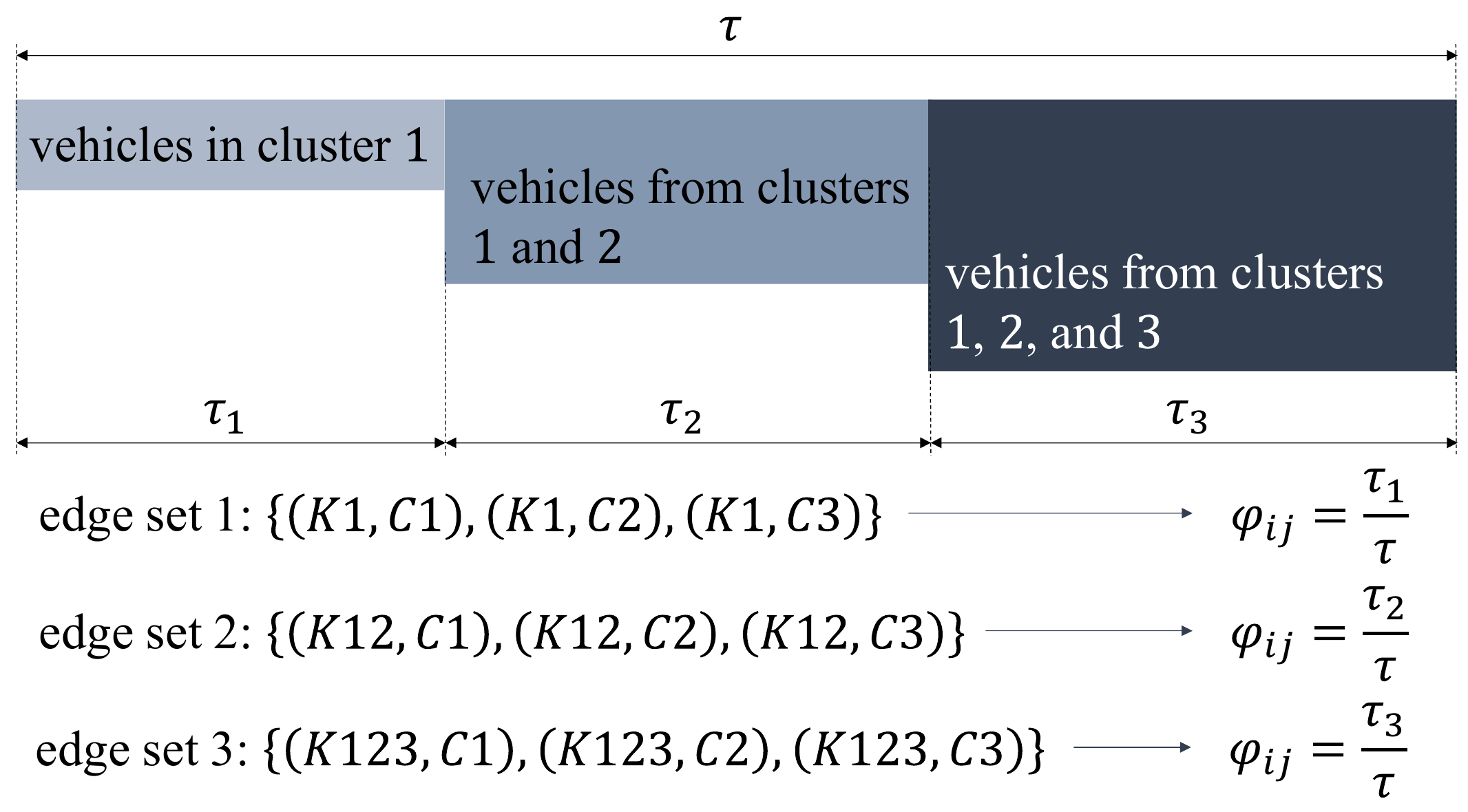}
\caption{Example of how the demand factors $\phi$ are calculated in cluster 1 of an instance with 3 clusters and decision epoch of length $\tau$.}
\label{fig:21}
\end{figure}

We further define the lower bounds of all vertices to zero and unbounded edge capacities as follows:

\begin{equation}
    l_{ij} = 0 \quad \forall (i,j) \in E \label{eq:17}
\end{equation}

\begin{equation}
    u_{ij} = \infty \quad \forall (i,j)\in E \label{eq:18}
\end{equation}

 We also define the balance vectors for the source and sink vertices $s$ and $t$ as follows:

\begin{equation}
    b(i) = S_i \quad \forall i \in A\label{eq:19}
\end{equation}

\begin{equation}
    b(t)= -\sum_{v \in A}b(v)=-\sum_{i \in A} S_i \label{eq:20}
\end{equation}

Parameter $S_i$ in \eqref{eq:19} and \eqref{eq:20} denotes the available vehicles $S_i$ in cluster $n(i)$ at the start of the current period.

As such, the vehicle redistribution problem introduced in section \ref{RAP} can be solved using the following nonlinear minimum cost flow optimization problem:

\textit{Model 1}:
\begin{mini!}[1]{}{\sum_{(ij) \in E} c_{ij} (x_{ij}) \label{eq:21a}}{}{}{}{}
    \addConstraint{x_{ij}}{\geq l_{ij}}{\quad}{\forall (i,j) \in E}{\label{(21b)}}
    \addConstraint{x_{ij}}{\leq u_{ij}}{\quad}{\forall (i,j) \in E}{\label{(21c)}}
    \addConstraint{b(j)}{=\sum_{i:ji \in E} x_{ji} - \sum_{i:ij \in E} x_{ij}}{\quad}{\forall j \in V}{\label{(21d)}}
    \addConstraint{x_{ij}}{\in \mathbb{R}}{\quad}{\forall (i,j) \in E}{\label{(21e)}}
\end{mini!}

Equations \eqref{(21b)}-\eqref{(21c)} ensure the flow of vehicles through each edge $(i,j)$ is within the specified lower and upper bounds respectively, as specified in equations \eqref{eq:17} and \eqref{eq:18}. \eqref{(21d)} is the flow continuity constraint as specified in equation \eqref{eq:6}.

\subsection{Convex Minimum Cost Flow Transformation} \label{CMCF}
The objective function in \eqref{eq:21a} is nonlinear, as a result of the costs for edges in $\mathcal{A}$ and $\mathcal{B}$. Such cost functions include the term $N^t_{ij}(x)$ introduced in \eqref{eq:4}, of which its component $q^t_{ij}(x)$ (eq. \eqref{eq:3}) and sub-component $w^t_{ij}(x)$ (eq. \eqref{eq:2}) are nonlinear. Identifying the nature of equation \eqref{eq:15} is paramount for the choice of a solution method for \textit{Model 1}. Classifying equation \eqref{eq:15} would suffice as \eqref{eq:16} is a linear sum of equation \eqref{eq:15} instances.

According to \cite{karamanis2020abm}, equation \eqref{eq:2} will result in maximum wait time up to the point of critical fleet size, when the number of idle vehicles $I^t_{ij}(x)$ will be sufficient to sustain the rate of incoming requests. Including this notion in the cost, equations will result in increasing costs for trip edges up to the point of critical fleet size. Instead, to assist the solution process and reduce concavity, we choose the following relaxed version of equation \eqref{eq:2} which can result in similar wait times for fleet sizes above the critical fleet size:

\begin{equation}
    w^t_{ij}(x) = \frac{\alpha_i^t Z^t_{ij}}{x + 1} \quad \forall i,j \in J, t \in T, x \in \mathbb{R}^+ \label{eq:2 new}
\end{equation}



The parameter $\alpha_i^t$ can be determined by using linear least squares regression for specific values of $Z^t_{ij}$ and pairs of $x$ and $w^t_{ij}$ identified via ABM. To aid our analysis, we assume there is only one alternative ride-sourcing option which offers identical pricing rates and travel times, with a fixed wait time $\bar{w}$. We expect that the ride-sourcing platform prices its rides at a rate higher than its cost per time, such that $p>C_M$. We finally assume $Z^t_{ij}$ is large enough to justify the cost of redistribution. For notation convenience, we omit any index notation $i,j,t$ in the proof of the following theorem.

\begin{theorem}
$h_{ij}(x_{ij})$ has an absolute minimum point in $x_{ij} \in [0, \infty]$.
\end{theorem}
\begin{proof}
We first consider the limit of $h(x)$ (eq. \eqref{eq:15}) as $x \rightarrow \infty$. We note that $\lim_{x \rightarrow \infty} w(x) = 0$, therefore, $\lim_{x \rightarrow \infty} g(x)$ is equal to some constant value $-r(\bar{v}+p)$. Consequently for $x \rightarrow \infty$, $q(x)$ converges to some finite maximum probability $p_{max}$. Since the upper bound of $q(x)$ is $1$, $\lim_{x \rightarrow \infty} N(x)=Z$. It is therefore straightforward to deduce the following limit:
\begin{equation}
    \lim_{x \rightarrow \infty} h(x)=\infty \label{eq:22}
\end{equation}

We now consider the limit of $h(x)$ for $x \rightarrow 0^+$. For a large $Z$, $\lim_{x \rightarrow 0^+} w(x) = \alpha Z$, $\lim_{x \rightarrow 0^+} g(x) = -\bar{v}(\alpha Z + r) -pr$. Due to the exponential nature of $q(x)$, for small values of $x$ (i.e. $x=1$ for large $Z$), we have that $\lim_{x \rightarrow 0^+} q(x) = 0$ and consequently $\lim_{x \rightarrow 0^+} N(x) = 0$. Therefore we arrive to the following result:
\begin{equation}
    \lim_{x \rightarrow 0^+} h(x)=0^+ \label{eq:23} 
\end{equation}
Let us now explore the case of $q(x)=0.5$. For $q(x)=0.5$, $w(x)=\bar{w}$, therefore by rearranging the terms of $w(x)$, we have that $x=\frac{\alpha}{\bar{w}}Z - 1$. Since $q(x)=0.5$, equation \eqref{eq:15} results to $h(x)=-0.5Zpr + C_M r x$. Substituting $x$ with $\frac{\alpha}{\bar{w}}Z - 1$, and simplifying we have the following equation:

\begin{equation}
    h(x) = r\bigg(-0.5 Z p + C_M (\frac{\alpha}{\bar{w}}Z - 1)\bigg) \label{eq:24}
\end{equation}

We know that the fleet operator chooses $p$ such that $p > C_M$ and we can scale $\alpha$, such that $\alpha < \bar{w}$ and $-0.5Z > \frac{C_M}{p} (\frac{\alpha}{\bar{w}}Z -1)$. Thus, by implementing the above inequalities, for $q(x)=0.5$, the following relationship holds:

\begin{equation}
    h(x) < 0 \quad \forall x \in \mathbb{R}^+ | q(x)=0.5 \land \frac{C_M \alpha}{p \bar{w}} \leq 0.5
\end{equation}

Using equations \eqref{eq:22} and \eqref{eq:23}, and by showing that $h(x)$ is negative for some $x \in \mathbb{R}^+$, we conclude that $h(x)$ has an absolute minimum point in $x \in [0, \infty]$.

\end{proof}

\begin{corollary}
    $h_{ij}(x_{ij})$ is convex for some domain $x_{ij} \in [x'_{ij}, x^*_{ij}]$. Where $h_{ij}(x^*_{ij})$ is the absolute minimum value of $h_{ij}(x_{ij})$ for $x_{ij} \in \mathbb{R}^+$ and $x'_{ij}$ is the largest value of $x_{ij}$ such that $h(x'_{ij})$ is a non-stationary inflection point and $x'_{ij} < x^*_{ij}$.
\end{corollary}

Our aim is to identify and utilise the convexity of the domain $[x'_{ij}, x^*_{ij}]$ of each non-linear edge cost function to solve \textit{Model 1} as a convex minimum cost flow problem. In line with convexity, we replace the upper bounds $u_{ij}$ of each non-linear edge $(i,j)$ to the absolute minimum value $x^*_{ij}$. Therefore additional to equation \eqref{eq:18}, we introduce the following:

\begin{equation}
    u_{ij} = x^*_{ij} \quad \forall (i,j)\in \mathcal{A} \cup \mathcal{B} \label{eq:26}
\end{equation}

In a similar fashion, setting the lower bound $l_{ij}$ of any non-linear edge $(i,j)$ to the inflection point $x'_{ij}$, would restrain our non-linear cost functions to the convex domain. Nonetheless, we refrain setting lower bounds to our problem to avoid potential infeasibility of \textit{Model 1}. Instead, we split each non-linear cost function to a piece-wise one, with a linear part between $[0, x'_{ij}]$, and non-linear convex part between $[x'_{ij}, x^*_{ij}]$. We linearise $h_{ij}(x_{ij})$ between $[0, x'_{ij}]$ to avoid any concave parts of the cost misguiding our solution algorithm (Section \ref{sec: ALG}) towards non-optimal solutions.

As we will see in the next section (Section \ref{sec: ALG}), our proposed solution algorithm identifies optimal solutions of convex functions by incrementally moving from higher absolute values of the cost derivative $\frac{dh_{ij}(x_{ij})}{dx_{ij}}$ towards values where the derivative approaches zero  \big($\frac{dh_{ij}(x^*_{ij})}{dx_{ij}} = 0$\big). Consequently, we set the equation of the linearised part $ h^L_{ij}(x_{ij})$ between $[0, x'_{ij}]$ for $h_{ij}(x_{ij})$ to the equation of the tangent of $h_{ij}(x_{ij})$ at $x'_{ij}$ follows:

\begin{equation} 
    h^L_{ij}(x_{ij}) =  \frac{dh_{ij}(x'_{ij})}{dx_{ij}} x_{ij} - \frac{dh_{ij}(x'_{ij})}{dx_{ij}} x'_{ij} + h_{ij}(x'_{ij}) \label{eq:27}
\end{equation}

Similarly, by performing the same procedure for $ f_{ij}(x_{ij})$ in equation \eqref{eq:16}, the linearised part of the cost has the following form:

\begin{equation} 
    f^L_{ij}(x_{ij}) =  \frac{df_{ij}(x'_{ij})}{dx_{ij}} x_{ij} - \frac{df_{ij}(x'_{ij})}{dx_{ij}} x'_{ij} + f_{ij}(x'_{ij}) \label{eq:28}
\end{equation}

We thereby introduce the following convex cost functions to replace equations \eqref{eq:10} and \eqref{eq:13} with equations \eqref{eq:29} and \eqref{eq:30} respectively:
\begin{equation}
    c_{ij}(x_{ij}) =  h^C_{ij}(x_{ij}) \quad \forall (i,j) \in \mathcal{A} \label{eq:29}
\end{equation}

\begin{equation}
    c_{ij}(x_{ij}) =  f^C_{ij}(x_{ij})  \quad \forall (i,j) \in \mathcal{B} \label{eq:30}
\end{equation}

The functions $h^C_{ij}(x_{ij})$ and $f^C_{ij}(x_{ij})$ in equations \eqref{eq:29} and \eqref{eq:30} respectively have the following form:

\begin{equation} 
    h^C_{ij}(x_{ij}) = \left\{
        \begin{array}{ll}
            h^L_{ij}(x_{ij}) & \quad x_{ij} \leq x'_{ij} \\
            h_{ij}(x_{ij}) & \quad x_{ij} > x'_{ij}
        \end{array}
    \right. \label{eq:31}
\end{equation}

\begin{equation} 
     f^C_{ij}(x_{ij}) = \left\{
        \begin{array}{ll}
            f^L_{ij}(x_{ij}) & \quad x_{ij} \leq x'_{ij} \\
            f_{ij}(x_{ij}) & \quad x_{ij} > x'_{ij}
        \end{array}
    \right. \label{eq:32}
\end{equation}

As such, by replacing equations \eqref{eq:10} and \eqref{eq:13} with equations \eqref{eq:29} and \eqref{eq:30} respectively and adapting the upper bounds of equation \eqref{eq:26} for non-linear edges, \textit{Model 1} becomes a Convex Minimum Cost Flow (CMCF) optimization problem.

\subsection{Edge-Splitting Pseudo-Polynomial Algorithm} \label{sec: ALG}
It was previously mentioned that it is possible to transform our vehicle redistribution problem into a CMCF problem. The flow $x_{ij}$ is a discrete quantity in our model as it considers the count of vehicles in each link $(i,j)$. Linear Minimum cost flow models adhere to the following theorem, as stated in \cite{Ahuja1993}:

\begin{theorem} \label{integrality}
(Integrality Theorem) If the capacities of all edges and the balance values of all the nodes are integer, the linear minimum cost flow problem always has an integer optimal flow.
\end{theorem}

For proof of the above theorem, we refer readers to \cite{Ahuja1993}. We thus deduce by restricting the parameters of the problem to integers (capacities and balance vectors), we can solve the linear minimum cost flow problem in polynomial time. We aim to exploit the integrality theorem, using an appropriate linearisation technique, to solve the CMCF problem transformation of \textit{Model 1} in polynomial time.

The CMCF problem has been previously tackled efficiently in the literature. \cite{minoux1984polynomial} initially proposed an extension of the scaling method for linear minimum cost flows presented in \cite{edmonds1972theoretical}, for convex cost flows with quadratic functions. At a subsequent stage, \cite{minoux1986solving} and \cite{hochbaum1990convex} separately conducted studies on solving minimum cost flows with general convex objectives. A variant of the algorithm proposed by \cite{minoux1986solving}, and along the lines of \cite{hochbaum1990convex}, is featured in \cite{Ahuja1993}. \cite{orlin2013fast} proposed polynomial algorithms for solving the CMCF problem in circles, lines or trees. The problem of quadratic CMCF was also tackled more recently in \cite{Vegh2016}, using an enhanced version of \cite{minoux1986solving}, utilizing the technique for linear minimum cost flows proposed in \cite{orlin1993faster}.

A consistent assumption of the studies which efficiently address the CMCF problem is that the edge costs are non-negative. Nonetheless, as we have seen in sections \ref{NMCF} and \ref{CMCF}, the non-linear cost functions of \textit{Model 1} can have negative values. The notion of negative costs (i.e. profit), implies that in an optimal solution of \textit{Model 1}, there would be edges with negative costs. As such, we refrain from using the above algorithms, and instead, we incorporate a modified version of the pseudo-polynomial algorithm for CMCF presented in \cite{Ahuja1993}. 

Our algorithm applies piece-wise linearisation by introducing parallel linear edges for each non-linear edge in the network. To limit the amount of additional parallel edges in the network, we start the algorithm with only two parallel edges for each non-linear edge $(i,j)$. As observed in equations \eqref{eq:31} and \eqref{eq:32}, $h^C_{ij}(x_{ij})$ and $f^C_{ij}(x_{ij})$ are partly linearised (i.e. for $x \leq x'_{ij}$). Consequently, by replacing $h_{ij}(x_{ij})$ and $f_{ij}(x_{ij})$ with their linearised versions between $x'_{ij}$ and $x^*_{ij}$, we initiate our algorithm with two parallel linear edges for each non-linear edge of \textit{Model 1}. Figure \footnote{For figure \ref{fig:3} we used the following parameters: $\bar{v}=0.3$, $\alpha=2$, $\bar{w}=5$, $p=1$, $C_M=0.3$.} \ref{fig:3} outlines the initial linearisation of costs for non-linear edges $(i,j) \in \mathcal{A}$. The linearisation process for edges $(i,j) \in \mathcal{B}$ is similar.

\begin{figure}[H]
\centering
\includegraphics[width=0.48\textwidth]{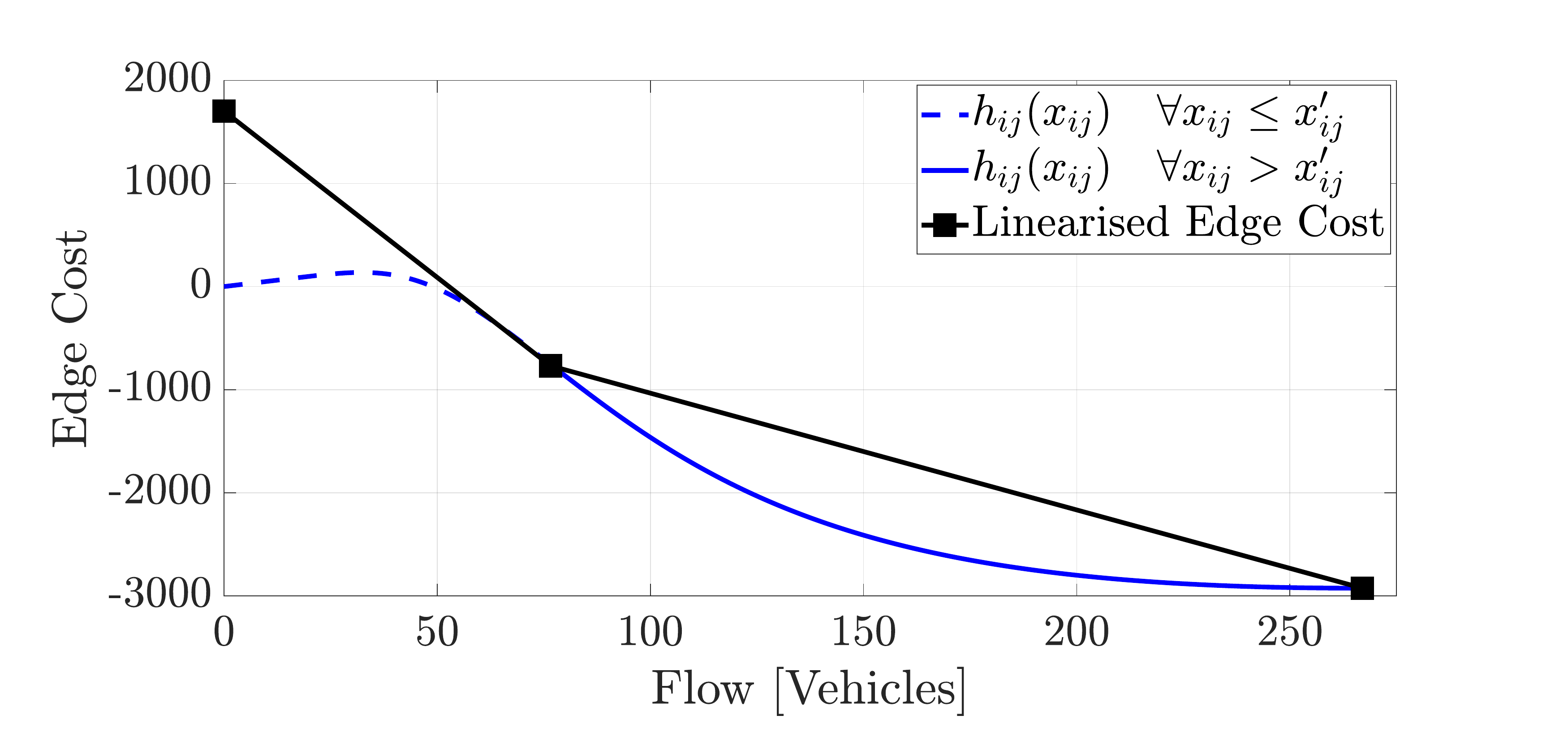}
\caption{Cost Function Variation for Edges $(i,j) \in \mathcal{A}$.}
\label{fig:3}
\end{figure}
  
As observed in figure \ref{fig:3}, due to convexity, the slope of each parallel linear edge from left to right gradually increases, from a minimum negative value to zero, as we move from $x_{ij}=0$ to $x_{ij}=x^*_{ij}$. Consequently, per unit flow is more expensive through the parallel linear edge corresponding to the non-linear section for $x_{ij} > x'_{ij}$. As such, assuming both parallel edges have residual capacity, flow through parallel edges with smaller slope is always prioritised over edges with larger slope value (i.e. closer to zero). This observation is described as the property of contiguity in \cite{Ahuja1993}.

Utilizing contiguity, for each parallel linear edge, we set the upper bound (capacity), to the difference between the right and left flow boundaries of the linearised edge, while maintaining zero lower bounds. Consequently, for our initial linearisation configuration, the upper bounds will be $u^1_{ij}=x'_{ij}-0$ and $u^2_{ij}=x^*_{ij} - x'_{ij}$ from left to right respectively, as observed in figure \ref{fig:3}. Since we initiate our algorithm with two parallel linear edges per non-linear edge, we denote their linear cost functions as $c^1_{ij}$ and $c^2_{ij}$ for edges corresponding to $x \leq x'_{ij}$ and $x_{ij} > x'_{ij}$ respectively. We show the transformation of each non-linear edge to a pair of linearised ones in figure \ref{fig:5}.

\begin{figure}[H]
\centering
\includegraphics[width=0.48\textwidth]{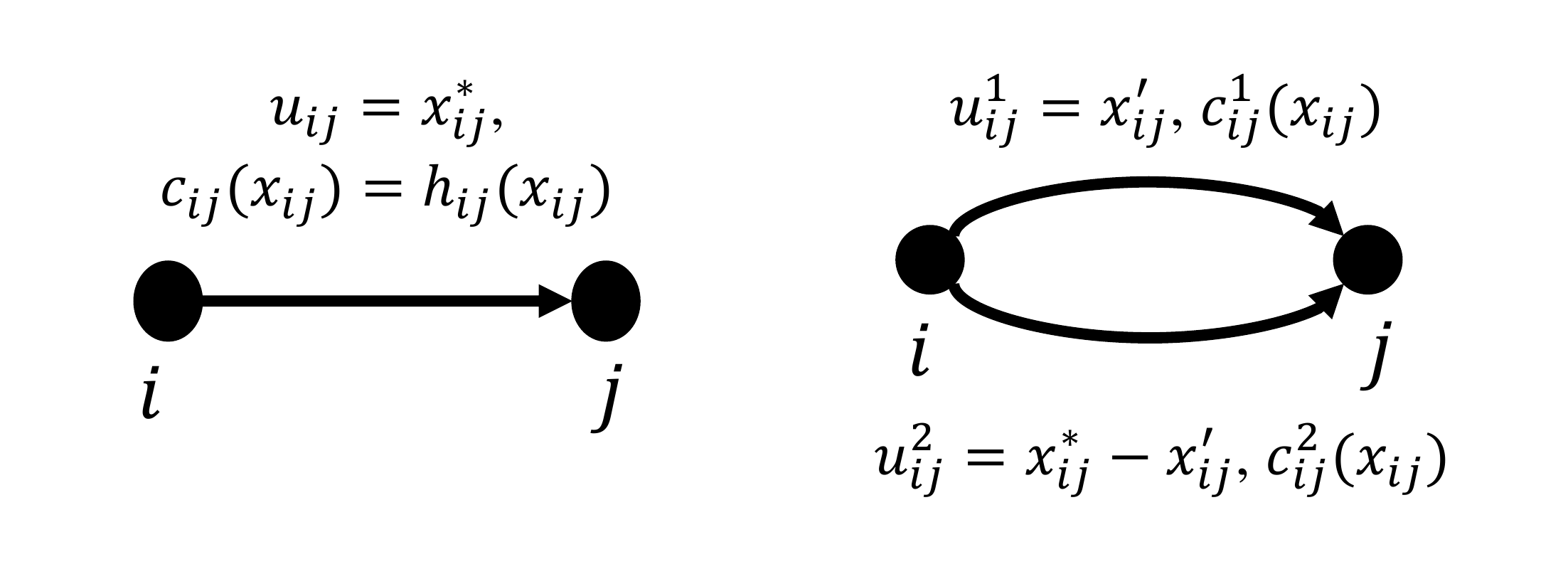}
\caption{Linear transformation of a non-linear edge $(i,j)$ to linear edges $(i,j)^1$ and $(i,j)^2$ with resulting linear costs $c^1_{ij}$ and $c^2_{ij}$ respectively.}
\label{fig:5}
\end{figure}

To maintain the validity of Theorem \ref{integrality}, we consider integer versions of $x'_{ij}$ and $x^*_{ij}$, such that the resulting capacities $u^1_{ij}$ and $u^2_{ij}$ are also integers. Furthermore, to identify point $x'_{ij}$ for each non-linear edge in \textit{Model 1}, we approximate the first and second derivatives of the non-linear cost using the forward and central difference formulas respectively. 

We introduce function $P(i,j)$ which identifies the set of parallel edges between any pair of vertices $i,j \in V$. We also refer to the linearised version of $G=(V, E)$ as $G_L=(V, E_L)$. Our non-linear edge linearisation of \textit{Model 1} up to this point can be described via the following formulation:

\textit{Model 2}:
\begin{mini!}[1]{}{\sum_{(ij) \in E}\sum_{k=1}^{P(i,j)} c^k_{ij} (x^k_{ij}) \label{eq:33a}}{}{}
    \addConstraint{x^k_{ij}}{\geq l^k_{ij}\qquad\forall (i,j) \in E,\forall k \in P(i,j)}{}{\label{(33b)}}
    \addConstraint{x^k_{ij}}{\leq u^k_{ij}\qquad\forall (i,j) \in E,\forall k \in P(i,j)}{}{\label{(33c)}}
    \addConstraint{b(j)}{=\sum_{i:ji \in E}\sum_{k=1}^{P(j,i)} x^k_{ji} - \sum_{i:ij \in E}\sum_{k=1}^{P(i,j)} x^k_{ij}\quad\forall j \in V}{}{\label{(33d)}}
    \addConstraint{x^k_{ij}}{\in \mathbb{R}\qquad\forall (i,j) \in E,\forall k \in P(i,j)}{}{\label{(33e)}}
\end{mini!}

A high level structure of our CMCF algorithm is outlined in Algorithm \ref{CMCF Alg}. Algorithm \ref{CMCF Alg} utilises graphs $G$ and $G_L$, and the structure of \textit{Model 2} to identify the set of minimum cost flows $F$, such that $x^k_{ij} \in F$, with $(i,j)^k \in E_L$. Algorithm \ref{CMCF Alg} incorporates an iterative procedure to arrive at the optimal solution for the CMCF of graph $G$.

\begin{algorithm}
\caption{CMCF Edge-Splitting Algorithm} \label{CMCF Alg}
\begin{algorithmic}[1]
\State Inputs: \textit{Model 2}, graph $G_L$ and graph $G$
\State $OPT =$ false
\While{$OPT =$ false}
\State $F \gets \emptyset$
\State $F = NetworkSimplex(F, G_L, \textit{Model 2})$
\State $U = SplittableEdges(F, G_L)$
\If{$U=\emptyset$}
\State $OPT=$ true
\Else
\For{$(i,j)^k \in U$}
\State $G_L = Split(G, G_L, (i,j)^k)$
\EndFor
\EndIf
\EndWhile
\State Output: Flow $F$
\end{algorithmic}
\end{algorithm}

To facilitate our edge-splitting algorithm, we define the boolean variable $OPT$ with the default value set to false, which signals the algorithm to stop if an optimal solution is found during an iteration (i.e. if $OPT$ is true). At the beginning of each iteration, we use the network simplex algorithm \cite{Ahuja1993} to solve the minimum cost flow problem and obtain a set $F$ of flows $x^k_{ij}$. We thereby screen through flows $x^k_{ij} \in F$ using the function $SplittableEdges(F, G_L)$, to identify the set $U$ of linearised parallel edges which are subject to further splitting. If $U$ is an empty set, the set $F$ of flows $x^k_{ij}$ is an optimal solution to the CMCF version of \textit{Model 1}. Otherwise, we proceed with splitting each edge $(i,j)^k \in U$ using the routine $Split(G, G_L, (i,j)^k)$, which updates the linearised graph $G_L$ and move to the next iteration. 

We use the function $NetworkSimplex(F, G_L, \textit{Model 2})$ to denote the procedure of solving \textit{Model 2} using network simplex and populating set $F$. We omit presentation of the network simplex algorithm as it is well known in the literature. The routine followed for the $SplittableEdges(F, G_L)$ function is outlined in Algorithm \ref{Splittable}. As observed in Algorithm \ref{Splittable}, we investigate the flow of each parallel edge $(i,j)^k$ of graph $G_L$ which exists in the convex domain of edge $(i,j)$ of graph $G$ (i.e. $k>1$). If the edge $(i,j)^k$ is the first edge in $(i,j)$ which is not at capacity and can be divided, we add edge $(i,j)^k$ to the set $U$.

The $Split(G, G_L, (i,j)^k)$ function used in Algorithm \ref{CMCF Alg} is outlined in Algorithm \ref{Split}. Initially, we identify the corresponding upper ($x_U$) and lower ($x_L$) flow values of edge $(i,j)^k$ in the convex domain of edge $(i,j)$. Since the flow $F$ is contiguous, we can find $x_U$ and $x_L$ by finding the total flow in $(i,j)$, and the total flow in $(i,j)$ excluding $x^k_{ij}$ respectively. We then identify the split point $x_S$ as the midpoint\footnote{$\ceil[big]{X}$ denotes the ceiling function.} of $x_U$ and $x_L$. We thus remove the edge $(i,j)^k$ and append $P(i,j)$ with the indices of the two new parallel edges to be added. For each new parallel edge, we find its upper and lower bound, as well as its cost function in a similar fashion as described earlier in the construction of $G_L$ before initialising Algorithm \ref{CMCF Alg}. We conclude the split function by adding the two new parallel edges in $E_L$ and returning the updated graph $G_L$.

\begin{algorithm}
\caption{SplittableEdges Function} \label{Splittable}
\begin{algorithmic}[1]
\State Inputs: Set $F$ of flows $x^k_{ij}$, graph $G_L$
\State $U \gets \emptyset$
\For{$(i,j) \in ((K \times C) \cup (D \times t))$}
\For{$k \in P(i,j)\setminus k=1$}
\If{$(l^k_{ij} \leq x^k_{ij} < u^k_{ij}) \cup (x^{k-1}_{ij} = u^{k-1}_{ij})$}
\If{$(u^k_{ij} > 1)$}
\State $U \gets U \cup (i,k)^k$ 
\EndIf
\EndIf
\EndFor
\EndFor
\State Output: Set $U$ of splittable edges
\end{algorithmic}
\end{algorithm}

\begin{algorithm}
\caption{Split Function} \label{Split}
\begin{algorithmic}[1]
\State Inputs: Graph $G$, graph $G_L$ and edge $(i,j)^k$
\State $x_U = \sum_{n \in P(i,j)} x^n_{ij}$
\State $x_L = \sum_{n \in P(i,j) \setminus n=k} x^n_{ij}$
\State $x_{S} = \ceil[\big]{\frac{x^U_{ij} + x^L_{ij}}{2}}$
\State $E_L \gets E_L\setminus (i,j)^k$
\State $n = max(P(i,j))$
\State $P(i,j) \gets P(i,j) \cup (n+1) \cup (n+2)$
\State $u^{n+1}_{ij} = x_S - x_L$
\State $u^{n+2}_{ij} = x_U - x_S$
\State $l^{n+1}_{ij} = 0$
\State $l^{n+2}_{ij} = 0$
\State Define $c^{n+1}_{ij}(x)$  by finding the linear equation between points $[x_L, c_{ij}(x_L)]$ and $[x_S, c_{ij}(x_S)]$.
\State Define $c^{n+2}_{ij}(x)$ by finding the linear equation between points $[x_S, c_{ij}(x_S)]$ and $[x_U, c_{ij}(x_U)]$.
\State $E_L \gets E_L \cup (i,j)^{n+1} \cup (i,j)^{n+2}$
\State Output: Updated linearised graph $G_L$
\end{algorithmic}
\end{algorithm}

The rationale behind our solution method, as described in Algorithms \ref{CMCF Alg}-\ref{Split}, is that we keep splitting linearised edges in the convex domain until we satisfy some optimality conditions. Specifically, we obtain the optimal solution when all the flows in each of the linearised edges in this domain are either zero or equal to the upper bound, and no further splitting can induce incremental cost savings. 

If the total input flow (i.e. $\sum_{i \in A} S_i$) is large enough, Algorithm \ref{CMCF Alg} allocates the upper bound flow in each of the non-linear edges of $G$ due to their negative costs (profitable edges). The case described above would terminate after the first iteration with the optimal solution of the CMCF version of \textit{Model 1}. Otherwise, for each non-linear edge, the first parallel edge which is not at capacity is split into two parts. For any split edge, due to convexity, the flow in the next iteration would always be confined within the resulting pair of parallel edges.

As a result of the above description, Algorithm \ref{CMCF Alg} terminates when for each non-linear edge, the splitting procedure produces parallel edges of unit capacity (i.e. $u^k_{ij}=1$). Consequently, the number of iterations is logarithmic and relates to the maximum interval $x^*_{ij} - x'_{ij}$ out of all the non-linear edges. If we denote this maximum interval as $\Delta$, we need to solve the minimum cost flow problem $\log_2(\Delta)$ times, adding at most $|J|^3 + |J|$ parallel edges to $G_L$ at each iteration. 

We can express the cardinality $m$ of the set of edges $E$ in terms of $|J|$; hence each network simplex run is polynomially bounded by the number of variables of the original problem. However,  we cannot express $\Delta$ by the number of variables in the network $G$. Consequently, Algorithm \ref{CMCF Alg} runs in pseudo-polynomial time. Nonetheless, even in extreme practical cases $\log_2(\Delta)$ is a small number (i.e. for $\Delta=10000$, $\log_2(\Delta)\approx 13$), hence our algorithm can be applied in practical implementations.

\section{Discussion} \label{sec: III}
\subsection{Agent-Based Model}
We tested the effectiveness of our redistribution algorithm in a simulated ride-sourcing environment using the ABM framework specified in \cite{karamanis2020abm} with a first-in-first-out (FIFO) customer assignment policy. Customer choices in the simulator were decided on an individual level using equations \eqref{eq:1} and \eqref{eq:3} with individual wait and travel times identified during the simulation. Additional to the framework in \cite{karamanis2020abm}, we included an empty relocation state for vehicles, during which vehicles can still be assigned to customers.

We selected the area of Manhattan, NYC to apply a case study of the algorithm due to the comprehensive trip data-set available in \cite{TLC2019} which served as our demand input. Travel times in the network were calculated using the OSMnx library \cite{BOEING2017126}. By assuming a small proportion of traffic attributes to ride-sourcing, we omitted endogenous congestion in our ABM. Nonetheless, we accounted for exogenous congestion by applying a 20\% penalty to the free-flow speeds in residential and motorway link segments, and 40\% elsewhere during peak hours.

Using the data-set in \cite{TLC2019}, we created typical demand profiles for weekdays in Manhattan, NYC  from 05:00 am to 12:00 am. K-means clustering was used to split the road-network into twenty clusters, with a universal value of $\alpha$ (equation \eqref{eq:2 new}) for convenience. To benchmark the performance of our algorithm, we tested it against the case of no redistribution and also against a linear programming (LP) method from the state-of-the-art which does not assume any supply-demand elasticity, namely the rebalancing model in \cite{wallar2018vehicle}. We selected the study in \cite{wallar2018vehicle} as it has been tested against other central rebalancing methods such as \cite{alonso2017demand}. All algorithms were tested using different fleet sizes from 2500 to 15000 vehicles.

For this study, we used UK estimates of the value of time $\bar{v}$ and vehicle moving costs $C_M$ from \cite{DfT2018} as information on values of time for New York was not available. Consequently, the average value of time $\bar{v}$ was set to 17.69 GBP/hour. The vehicle moving cost $C_M$ was set to the conservative estimate of 12.96 GBP/hour to reflect current driver valuations. The idle vehicle cost $C_I$ was set to 1 GBP/hour to account for parking costs, whereas the price per minute for a ride $p$ was set to 1.00 GBP/min to reflect previous research on AV pricing \cite{Karamanis2018}. The value of $\bar{w}$ was set to 5 (minutes). 

Our algorithm performs allocations based on demand expectations for two subsequent periods ($Z^1_{ij}$ and $Z^2_{ij}$). Since the application of predictive algorithms is beyond the scope of this paper, we assumed that the platform has complete knowledge of the demand in the two subsequent periods for each cluster when applying our algorithm and the model in \cite{wallar2018vehicle}. Nonetheless, to scrutinise how demand prediction accuracy affects the efficacy of our algorithm, we also considered an instance with a normally distributed error with mean $20\%$ and standard deviation $10\%$ on the total demand $Z^t_{ij}$. 

To investigate the robustness of our method and the effectiveness of the window length $\tau$, we created instances with different values of $\alpha$ and $\tau$, respectively. The tested instances and parameters are outlined in table \ref{tab:instances}. Where LP in table \ref{tab:instances}, refers to the model of \cite{wallar2018vehicle} and $a_3$ represents the version of \textit{Model 2} with a random error on total demand $Z^t_{ij}$.

\begin{table}[H]
\caption{Instances tested in the simulation.}
\centering
\label{tab:instances}
\begin{tabular}{@{}lllll@{}}
\toprule
ID    & Method  & Clusters & $\tau$ [min] & $\alpha$ \\ \midrule
$a_0$ & Model 2 & 20       & 30     & 0.75     \\
$a_1$ & Model 2 & 20       & 15     & 0.75     \\
$a_2$ & Model 2 & 20       & 30     & 1.5      \\
$a_3$ & Model 2 & 20       & 30     & 0.75     \\
$b_1$ & LP      & 20       & 30     &          \\
$b_2$ & LP      & 20       & 15     &          \\
$b_3$ & LP      & 20       & 5      &          \\ \bottomrule
\end{tabular}
\end{table}

\subsection{Performance Analysis}
We used four different Key Performance Indicators (KPIs) to compare the effectiveness of each set-up outlined in table \ref{tab:instances} with the case of no redistribution. Specifically, we used wait time reduction, market share improvement, profit improvement and additional mileage, expressed as percentage shifts from the no redistribution benchmark. 

As observed in figures \ref{fig:7}-\ref{fig:10}, we find that frequent redistribution degrades the efficacy of all tested methods. Our algorithm was found to be ineffective for $\tau = 15$ minutes, whereas the LP benchmark algorithm presented noticeable degradation for a window length of 5 minutes. This ineffectiveness is expected, as both algorithms use central optimization and a refined window length deteriorates the quality of aggregation. We, therefore, deduce that our algorithm is best fitted for low frequency in redistribution (i.e. every 30 minutes).

As observed in figure \ref{fig:7}, vehicle redistribution via our algorithm can reduce the average customer wait time up to more than 50\% when compared to no redistribution. Furthermore, when considering market share, our algorithm contributes to an increase of up to 10\%, as shown in figure \ref{fig:8}. However, as the fleet size increases to larger values, the contribution of redistribution towards higher market share and reduced wait times becomes marginal due to the abundance of available vehicles across the network.

\begin{figure}[h]
\centering
\includegraphics[width=0.48\textwidth]{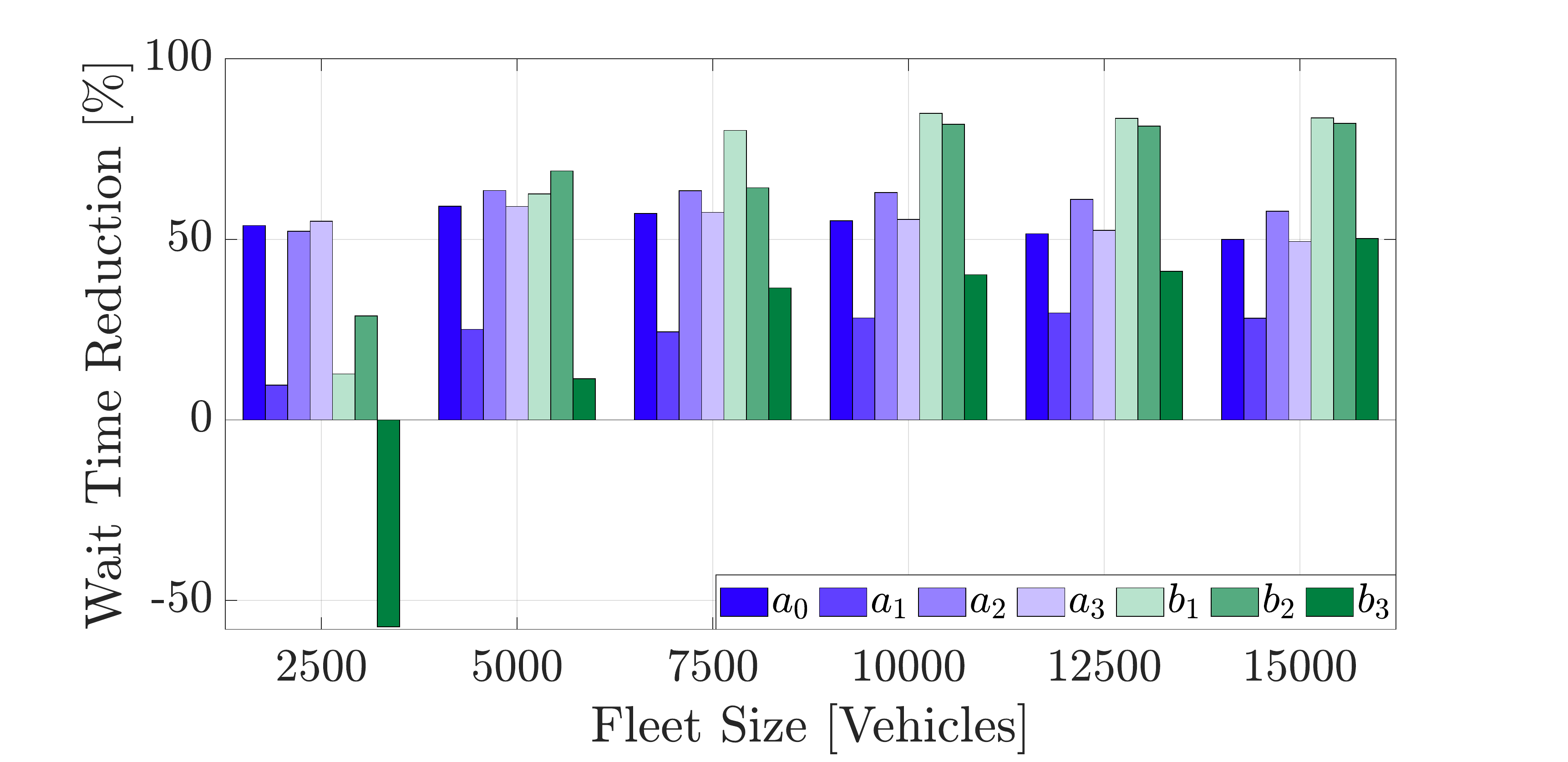}
\caption{Average wait time reduction of customers for different fleet sizes and redistribution strategies.}
\label{fig:7}
\end{figure}

\begin{figure}[h]
\centering
\includegraphics[width=0.48\textwidth]{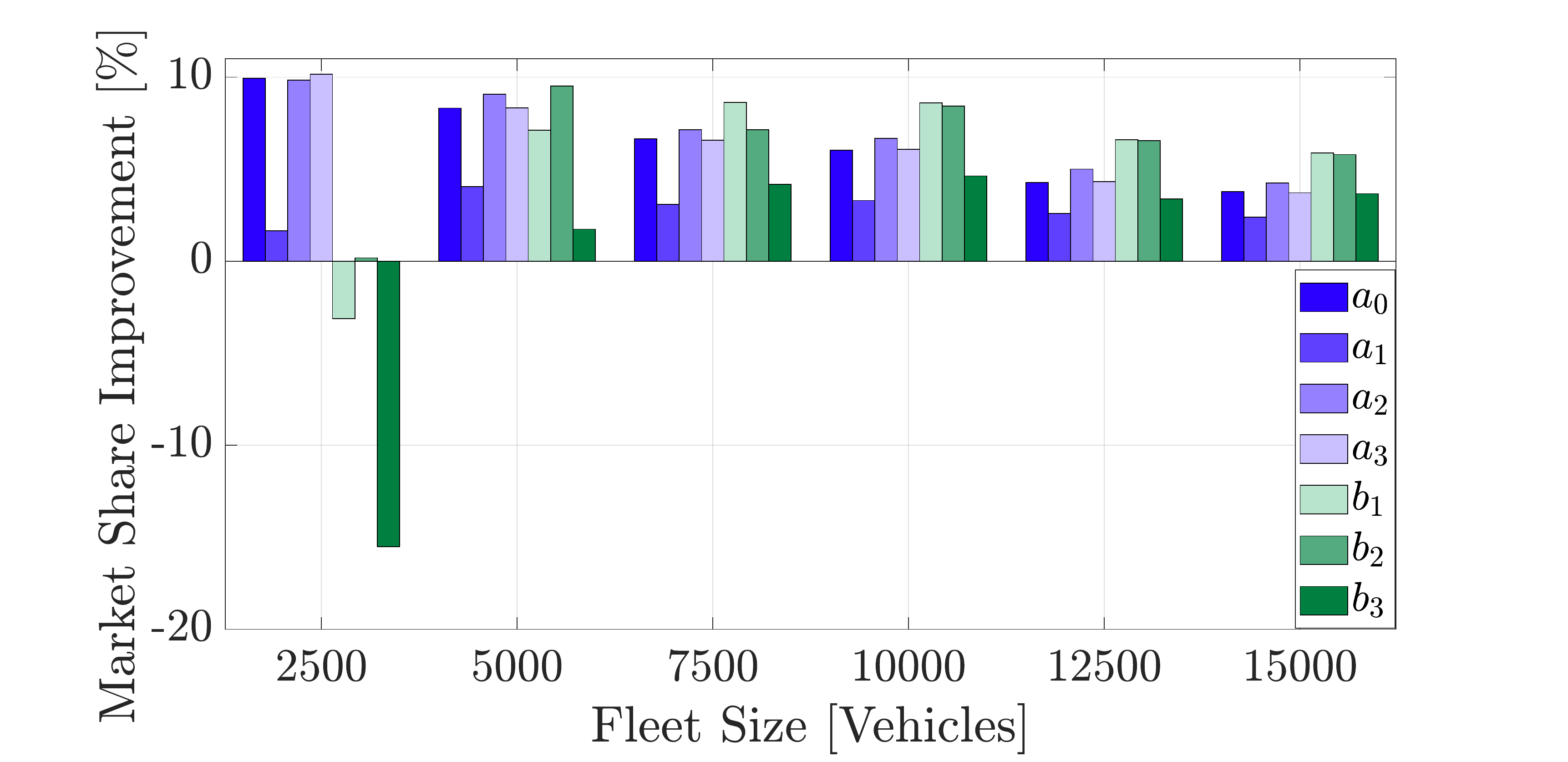}
\caption{Market share improvement for different fleet sizes and redistribution strategies.}
\label{fig:8}
\end{figure}

Comparing the results of our algorithm with the method proposed in \cite{wallar2018vehicle} and ignoring smaller window lengths ($a_1$ and $b_3$), we see that in many cases (above 2500 vehicles) the LP method outperforms \textit{Model 2} in reducing wait times and improving market share (figures \ref{fig:7} and \ref{fig:8}). Nonetheless, as our algorithm is modelled to minimise cost based on expected customer choices, the effectiveness of our algorithm over LP is clearly stated in profit improvement and additional mileage as observed in figures \ref{fig:9} and \ref{fig:10} respectively. 

Observing figures \ref{fig:9} and \ref{fig:10}, we see that our algorithm can achieve a sizeable profit increase of up to 10\% with an additional mileage of 20\% across the network. On the contrary, the LP algorithm achieves much lower wait times and higher market share in the expense of considerable additional mileage and a significant profit reduction (more than 10\% in most cases). This result is anticipated, as there is no notion of diminishing returns for additional redistribution in the LP method, which is the reason for non-linearity in our method. 

\begin{figure}[h]
\centering
\includegraphics[width=0.48\textwidth]{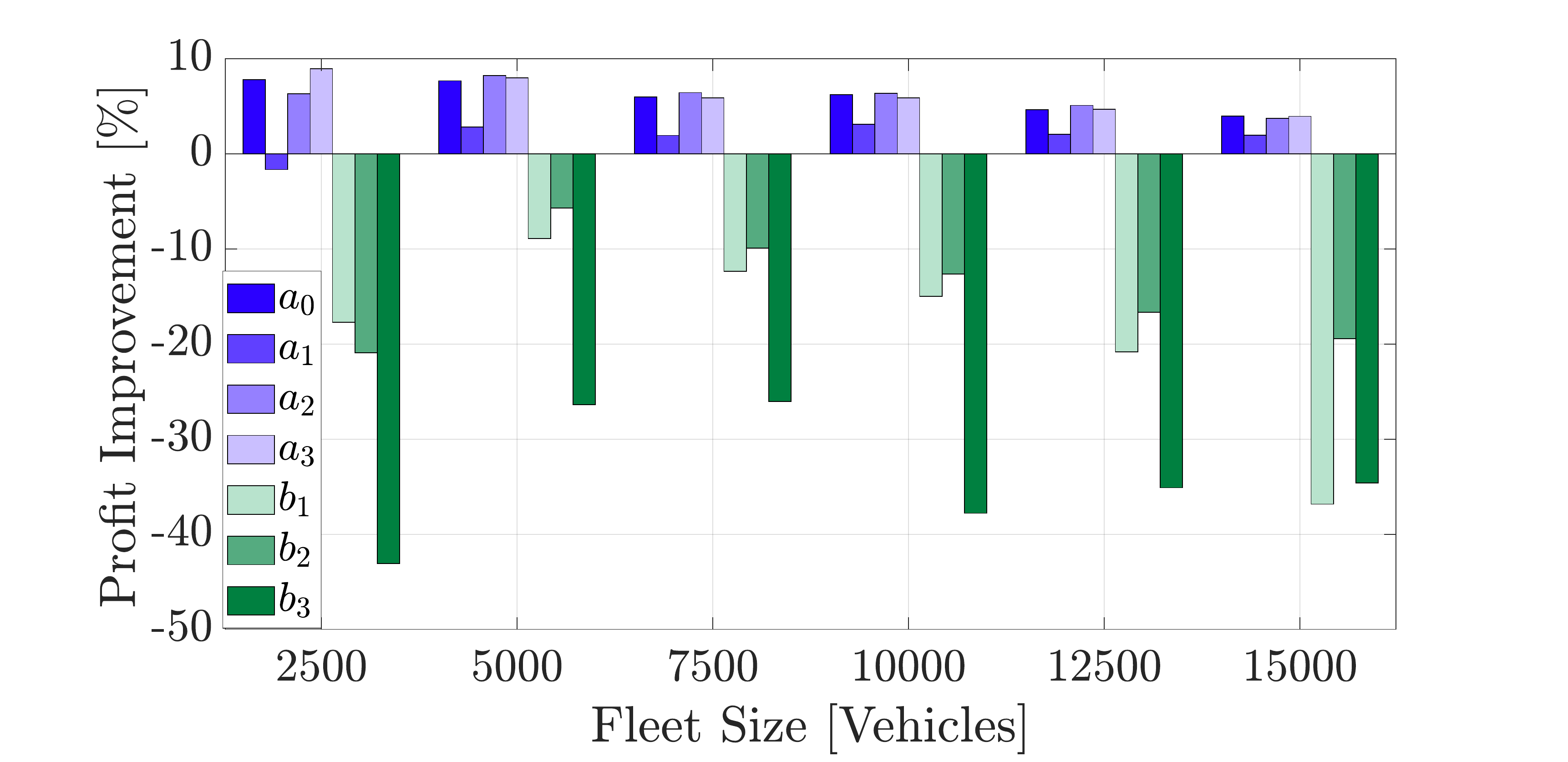}
\caption{Profit improvement for different fleet sizes and redistribution strategies.}
\label{fig:9}
\end{figure}

\begin{figure}[h]
\centering
\includegraphics[width=0.48\textwidth]{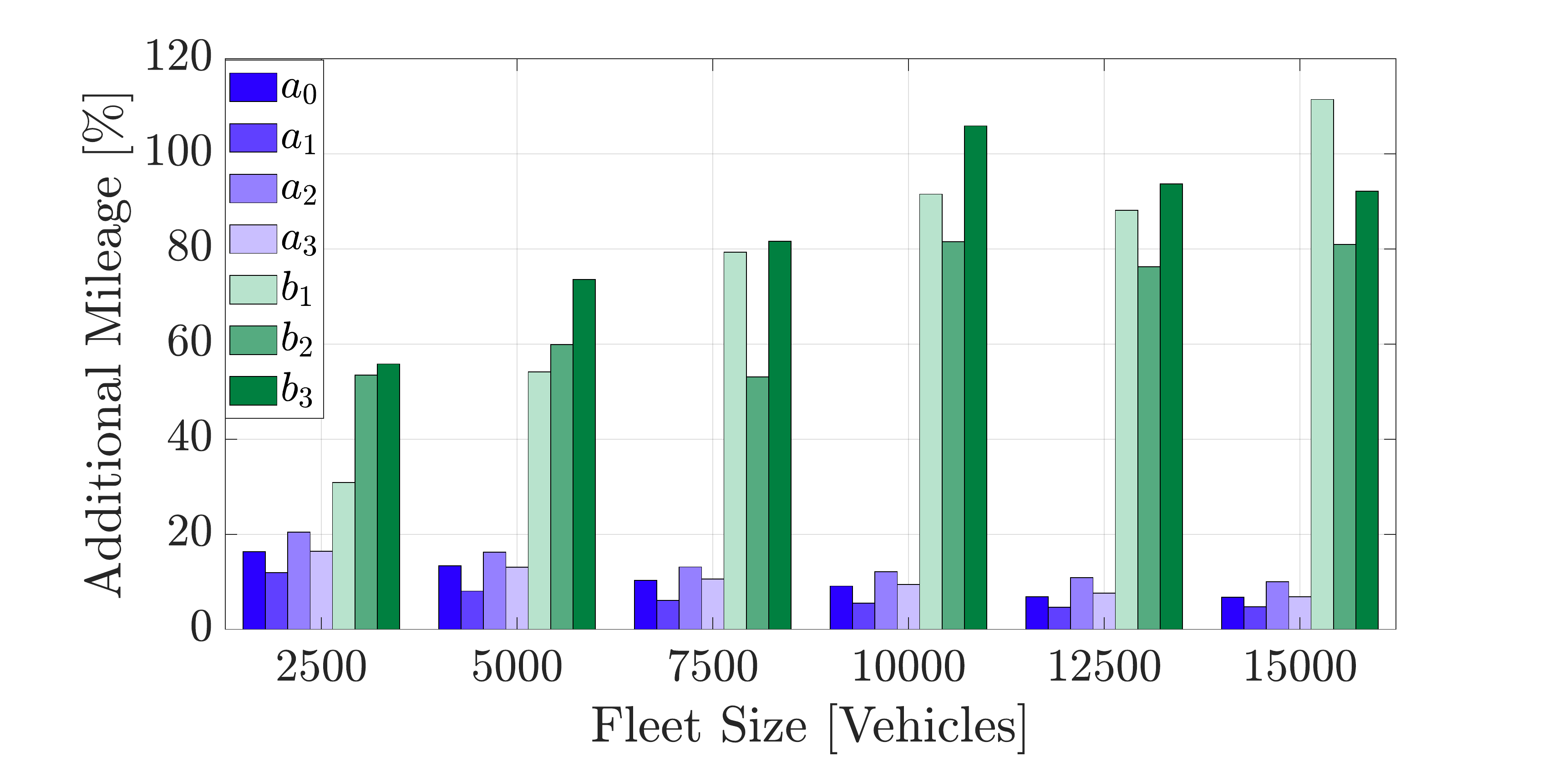}
\caption{Additional mileage for different fleet sizes and redistribution strategies.}
\label{fig:10}
\end{figure}

\subsection{Practical Implementation}
To assess the usefulness of \textit{Model 2} in practical implementations, we assess its robustness with regards to input functions such as equation \eqref{eq:2 new} and uncertainty in demand estimation. We see that varying $\alpha$ in equation \eqref{eq:2 new} from 0.75 to 1.50 does not have a significant effect on the resulting redistribution. As \eqref{eq:2 new} is only used by the operator to estimate the effect of vehicle supply on the average wait time, variations of the wait time function can change the capacity of trip links in \textit{Model 2}. Nonetheless, as the wait time change for either value is marginal when considering the large sizes of incoming demand and available vehicles in \eqref{eq:2 new}, so is the change across all KPIs when comparing instance $a_0$ to instance $a_2$.

Furthermore, when comparing the effect of a randomly distributed error on demand estimation (instance $a_3$), we see that the performance of $a_3$ is consistently similar to both $a_0$ and $a_2$ across all KPIs. This similarity is most likely present due to setting an upper bound on the capacity of trip links in \textit{Model 2} and allowing for idle vehicles in each cluster to be available for assignment in-between redistribution decisions.

Finally, to confirm the efficiency of Algorithm \ref{CMCF Alg}, which solves \textit{Model 2}, we measured its runtime for each fleet size in scenario $a_0$. We find that in the worst case, our algorithm requires an average of 74.7 seconds to identify the optimal redistribution, which we consider acceptable if a central fleet redistribution decision is made in windows of 30 minutes. The average runtimes for all tested fleet sizes are outlined in table \ref{tab:runtimes}.

\begin{table}[h]
\centering
\caption{Algorithm \ref{CMCF Alg} average runtimes for $a_0$ scenario and various fleet sizes.}
\label{tab:runtimes}
\begin{tabular}{@{}lllllll@{}}
\toprule
\begin{tabular}[c]{@{}l@{}}Fleet\\ Size\end{tabular}        & 2500 & 5000 & 7500 & 10000 & 12500 & 15000 \\ \midrule
\begin{tabular}[c]{@{}l@{}}Runtime\\ {[}sec{]}\end{tabular} & 74.7 & 23.8 & 23.8 & 24.2  & 24.1  & 23.8  \\ \bottomrule
\end{tabular}
\end{table}

\section{Conclusion} \label{sec: IV}

In this paper, we addressed the problem of vehicle redistribution in autonomous ride-sourcing markets to mediate supply-demand imbalance across a road-network. We used network theory to transform the vehicle redistribution problem into a CMCF problem with negative costs, accounting for customer behaviour under an assumed market structure. Our proposed edge-splitting algorithm solves the CMCF problem exactly in pseudo-polynomial time by allocating vehicles to spatio-temporal tasks. We demonstrated the practicability of our redistribution algorithm in an agent-based model simulating ride-sourcing in a large urban setting, such as Manhattan, NYC.

Our suggestions for future research in this area are several. First, we believe that transportation providers should quantify the effects of localised demand choice model structures in vehicle redistribution models. The underlying city/road-network structure can influence parameters such as wait time variation and cluster size. As such, studies which test redistribution algorithms on different networks could add value to current research. Finally, we note that our model considers optimizing the operations of privately owned fleets from an operator perspective. We thus believe that studies which test vehicle redistribution of publicly owned fleets from the perspective of welfare maximization could be useful.

\bibliographystyle{IEEEtran}
\bibliography{Bibliography}

\begin{IEEEbiography}[{\includegraphics[width=1in,height=1.25in,clip,keepaspectratio]{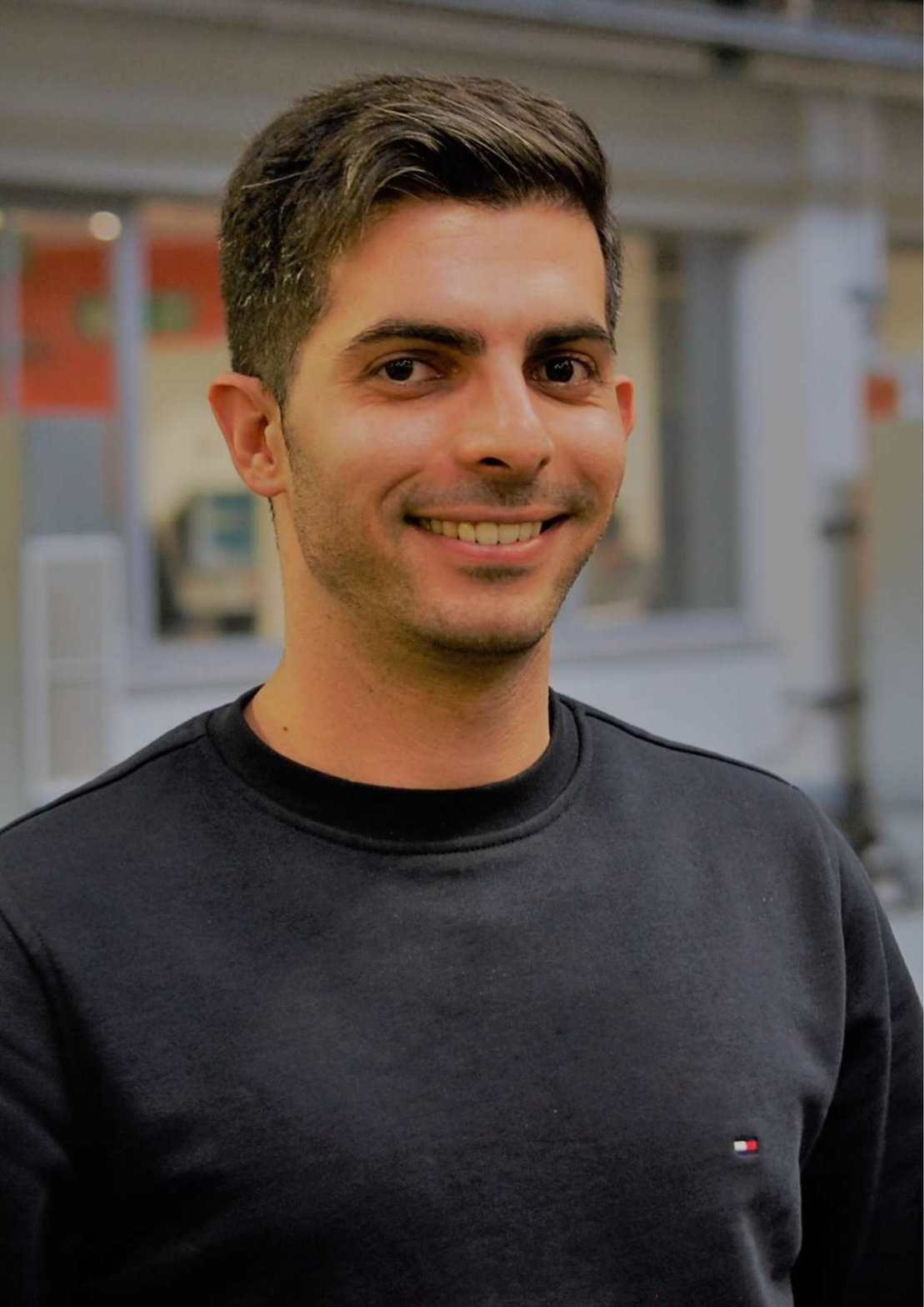}}]{Renos Karamanis}
received an M.Eng. degree in Civil and Environmental Engineering from Imperial College London in 2014. From 2014 to 2015 he worked as an engineering consultant for Mott MacDonald, a multi-disciplinary consultancy with headquarters in the United Kingdom. 

In 2015 he joined the Centre for Transport Studies in Imperial College London where he is currently a Ph.D. student. His research interests include developing operational research methods to improve the efficiency of pricing, assignment and vehicle redistribution of autonomous ride-sourcing fleets, and simulation modelling of ride-sourcing operations to aid policy suggestions.
\end{IEEEbiography}
\vskip -2\baselineskip plus -1fil
\begin{IEEEbiography}[{\includegraphics[width=1in,height=1.25in,clip,keepaspectratio]{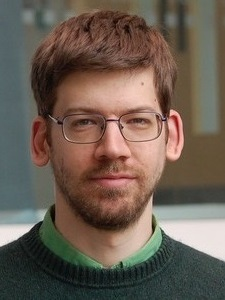}}]{Eleftherios Anastasiadis}
received a BSc in informatics and telecommunications from the University of Athens in 2011. He received an MSC and a PhD in theoretical computer science from the University of Liverpool in 2012 and 2017 respectively. 

In 2017 he worked as quality assurance engineer in Rosslyn Data Technologies Ltd. Since 2018 he is a postdoctoral researcher at the Transport Systems and Logistics Laboratory in the department of Civil and Environmental Engineering at ImperialCollege London. His research interests include approximation algorithms and mechanism design for network optimisation problems, and agent-based simulation for autonomous vehicle fleets. 
\end{IEEEbiography}
\vskip -2\baselineskip plus -1fil
\begin{IEEEbiography}[{\includegraphics[width=1in,height=1.25in,clip,keepaspectratio]{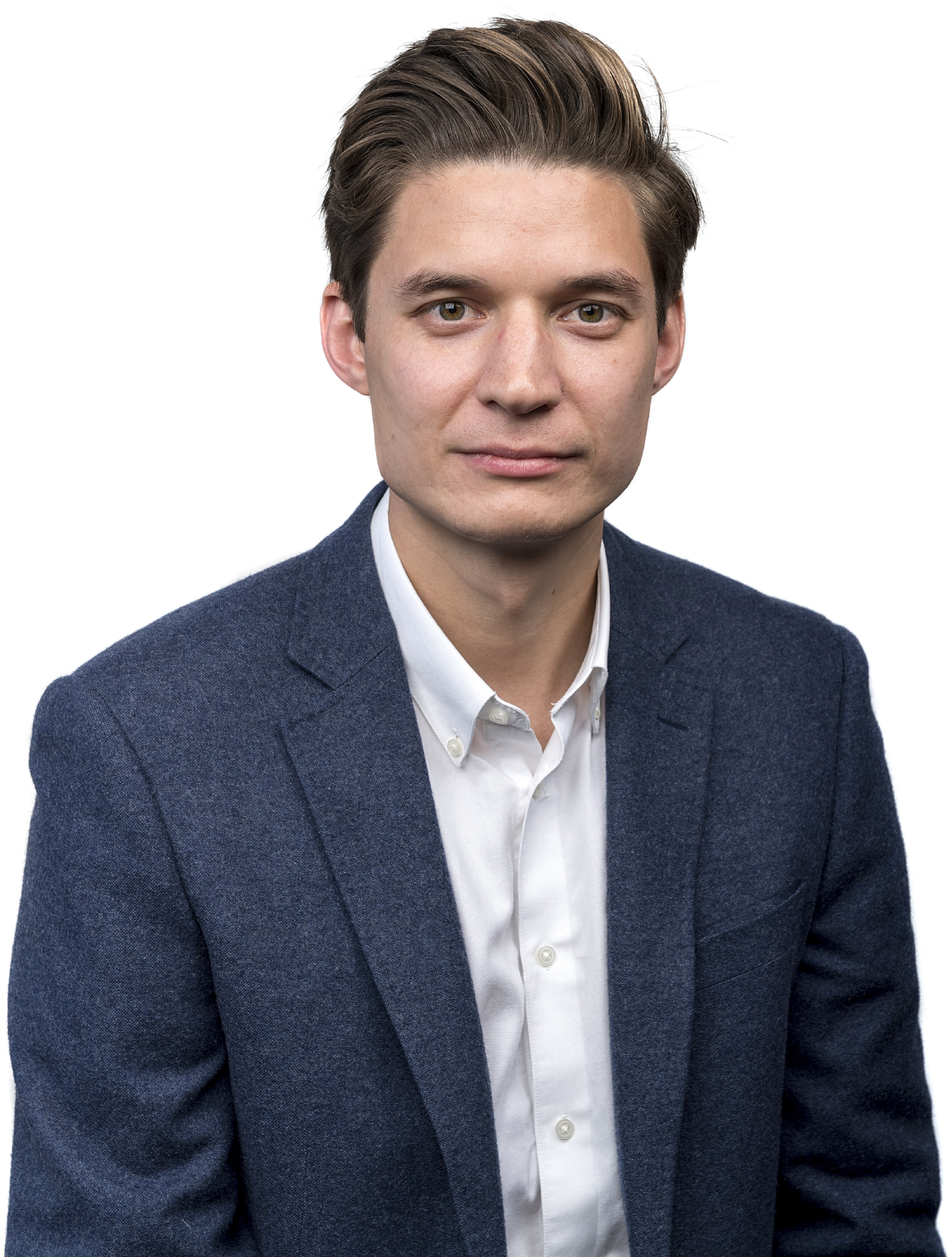}}]{Marc Stettler}
is a Senior Lecturer in Transport and the Environment in the Centre for Transport Studies and Director of the Transport \& Environment Laboratory. Prior to joining Imperial, Marc was a research associate in the Centre for Sustainable Road Freight and Energy Efficient Cities Initiative at the University of Cambridge, where he also completed his PhD.

His research aims to quantify and reduce environmental impacts from transport using a range of emissions measurement and modelling tools. Examples of recent research projects include: quantifying real-world vehicle emissions; using real-world vehicle emissions data to improve emissions models; evaluating economic and environmental benefits of Kinetic Energy Recovery Systems (KERS) for road freight; and quantifying aircraft emissions at airports. Marc is a member of the LoCITY ‘Policy, Procurement, Planning and Practice’ working group and the EQUA Air Quality Index Advisory Board. 
\end{IEEEbiography}
\vskip -2\baselineskip plus -1fil
\begin{IEEEbiography}[{\includegraphics[width=1in,height=1.25in,clip,keepaspectratio]{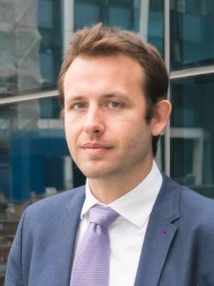}}]{Panagiotis Angeloudis}
is a Reader and Director of the Transport Systems and Logistics Laboratory, part of the Centre for Transport Studies and the Department of Civil \& Environmental Engineering at Imperial College London. He received an MEng in Civil \& Environmental Engineering in 2005 and a PhD in Transport Operations in 2009, both from  Imperial College London. 

His research interests lie on the field of transport systems and networks operations, with a focus on the the efficient and reliable movement of people and goods across land, sea and water. He was recently appointed by the UK Department for Transport to the Expert Panel for Maritime 2050 and was a member of the UK Government Office of Science Future of Mobility review team. He is affiliated with the Centre for Systems Engineering and Innovation, the Institute for Security Science and Technology, the Grantham Institute and the Imperial Robotics Forum.
\end{IEEEbiography}

\end{document}